\documentclass[12pt, a4paper]{article}%
\usepackage{amsmath,amssymb,eucal}
\usepackage{amsfonts}
\usepackage{mathrsfs}
\usepackage{slashed}
\numberwithin{equation}{section} 
\def\beq{\begin{equation}}
\def\eeq{\end{equation}}

\def\bC{ {{\mathbb{C}}}}
\def\bR{ {{\mathbb{R}}}}

\def\Tr{ {{\rm{Tr}}} }

\newcommand{\pd}{{\rm DP}}

\newcommand{\pk}[1]{p_{\kappa}}

\newcommand{\ol}{\overline{l}}
\newcommand{\ul}{\underline{l}}
\newcommand{\otau}{\overline{\tau}}

\newcommand{\oL}{\overline{L}}
\newcommand{\uL}{\underline{L}}
\newcommand{\on}{\overline{n}}
\newcommand{\un}{\underline{n}}

\newtheorem{defn}{{\bf Definition}}[section]
\newtheorem{thm}[defn]{{\bf Theorem}}

\newtheorem{rem}[defn]{{\bf Remark}}

\newtheorem{notation}[defn]{Notation}
\newenvironment{proof}[1][Proof]{\textbf{#1.} }{\hfill \rule{0.5em}{0.5em}}
\begin{document}

\title{From Loop Quantum Gravity to a Theory of Everything}
\author{Adrian P. C. Lim \\
Email: ppcube@gmail.com
}

\date{}

\maketitle

\begin{abstract}
Witten described how a path integral quantization of Wilson Loop observables will define Jones polynomial type of link invariants, using the Chern-Simons gauge theory in $\mathbb{R}^3$. In this gauge theory, a compact Lie group ${\rm G}$, together with a representation of its Lie Algebra $\mathfrak{g}$, describe the symmetry group and fundamental forces acting on the particles respectively.  However, it appears that this theory might be part of a bigger theory.

We will incorporate this theory into the Einstein-Hilbert theory, which when reformulated and quantized using a ${\rm SU}(2) \times {\rm SU}(2)$ gauge group, gives us a quantized theory of gravity in $\mathbb{R}^4$. In this theory, we can quantize area, volume and curvature into quantum operators. By using both the Chern-Simons and Einstein-Hilbert action, we will write down a path integral expression, and compute the Wilson Loop observable for a time-like hyperlink in $\mathbb{R}^4$, each component loop is coloured with a representation for the Lie Algebra $\mathfrak{g} \times [\mathfrak{su}(2) \times \mathfrak{su}(2)]$, unifying the fundamental forces with gravity. This Wilson Loop observable can be computed using link diagrams, and it can be written as a state model, satisfying a Homfly-type skein relations.

We will show that the Wilson Loop observable remains an eigenstate for the quantum operators corresponding to spin curvature, but it is not an eigenstate for the area and volume quantized operators, unless the representation for $\mathfrak{g}$ is trivial. This implies that in the Planck scale where quantum gravity is important, we see that all the particles are indistinguishable, hence the fundamental forces disappear and only interaction between matter and space-time remains.
\end{abstract}

\hspace{.35cm}{\small {\bf MSC} 2020: 83C45, 81V22, 81T20, 81T45} \\
\indent \hspace{.35cm}{\small {\bf Keywords}: Loop quantum gravity, theory of everything, Chern-Simons, \\
\indent \hspace{2.4cm} Einstein-Hilbert, linking number, area, volume, curvature, \\
\indent \hspace{2.4cm} topological field theory, link invariants}




\section{Introduction}\label{s.pre}

The Standard Model is very successful in unifying the electromagnetic, strong and weak forces, by describing these forces using a ${\rm U}(1) \times {\rm SU}(2) \times {\rm SU}(3)$ gauge theory. The gauge groups ${\rm U}(1)$, ${\rm SU}(2)$ and ${\rm SU}(3)$ describe the electromagnetic, weak and strong force respectively. See \cite{2004Jeff, Lindgren_2021} on the electromagnetic force and \cite{peskin1995} for weak and strong interactions. Although it is not explicitly stated, but a metric on the underlying ambient space $\bR^4$, is actually required to define the Standard Model Lagrangian. See \cite{Buchmuller:2006zu}.

The Standard Model does not include gravitons, the supposed force carrying particles for gravity, yet to be observed in nature. See \cite{Griffiths:111880}. Gravity is weak at energy levels way below the Planck's scale, compared to the other three forces. See \cite{Thiemann:2002nj}. Effects of gravity are neglected, thus one can use the Minkowski metric, or a flat metric. When the gravitational force has to be taken into consideration, the Minkowski metric has to be replaced with a different metric, solved from Einstein's equations. See \cite{schutz1985first}.

Unifying the 3 forces with gravity to obtain a Theorey of Everything has been the holy grail for physicists. Any attempt to unify gravity with the three fundamental forces has met with fierce resistance. One obvious reason is that we have yet to have a widely acceptable quantum theory of gravity. Smolin in \cite{Smolin2006-SMOTCF} argued that a correct quantum theory of gravity, must be independent of any background metric. Because the Standard Model is a theory dependent on a metric, hence a quantum theory of gravity is not compatible with the Standard Model.

As the theory of gravity is essentially a theory on the diffeomorphisms of space-time, a quantum theory of gravity is in fact asking for a quantum theory of space-time. Ashetekar referred it as quantum geometry in \cite{Ashtekar:2004vs}.  On the other hand, we defined quantum geometry, as a theory which describes how closed submanifolds in a four dimensional manifold, are `linked' together, up to time-like and time-ordered equivalence relation as defined in \cite{EH-Lim06}. These equivalence invariants take discrete values and include topological invariants.

There are many potential candidates for a quantum theory of gravity. See \cite{Smolin2006-SMOTCF}. Here, we would like to highlight Loop Quantum Gravity (LQG), which is a theory that is metric independent. See \cite{Thiemann:2002nj}. It can be reformulated as a gauge theory, using ${\rm G}_0 := {\rm SU}(2) \times {\rm SU}(2)$ gauge group, and the main computational object is an Einstein-Hilbert path integral given by Expression \ref{ex.eh.1}. This path integral computes the Wilson Loop observable, given by Equation (\ref{e.weh.1}), which is the average holonomy of $\mathfrak{su}(2) \times \mathfrak{su}(2)$-valued connections, translated over a time-like hyperlink in $\bR^4$. This average is over all possible connections, weighted by the exponential of the Einstein-Hilbert action. See \cite{EH-Lim02}.

The Wilson Loop observable is computed using the hyperlinking number between each component matter loop, with a geometric hyperlink, and we will henceforth term it as quantum Einstein-Hilbert invariant. The hyperlinking number will be defined in Definition \ref{d.c.2}. In the quantization of gravity using the Einstein-Hilbert path integrals, the hyperlinking number between each pair of component matter loops in a hyperlink, do not appear in the Wilson Loop observable. Neither do Homfly-type polynomial invariants show up in LQG. It appears that it might be an incomplete theory.

Physical observables such as area, volume and curvature can be quantized as operators in this theory via an Einstein-Hilbert path integral, computed from invariants in quantum geometry, independent of any metric. These will be summarized in Sections \ref{s.tqt} and \ref{s.av}. This will lend meaning to area, volume and curvature, without specifying a metric. As such, we have a discrete spectrum for all three quantities and we can even speak of a quanta of area or volume. This connects LQG with quantum geometry, justifying it as the quantization of space-time.

The main focus of LQG is to give a quantum theory of gravity, unlike string theory, which aims to unify all fundamental forces into a single theory. LQG itself is not a fundamental theory. See \cite{Rovelli1998}. Ashetekar in \cite{Ashtekar:2004vs} talked about the possibility of unifying quantum gravity with the fundamental forces. The purpose of this article then, is to incorporate gauge theories which govern the strong, weak and electromagnetic forces, into LQG. It is not possible to reformulate LQG, using ${\rm U}(1) \times {\rm SU}(2) \times {\rm SU}(3)$ gauge group. To unify the three forces with gravity, we need to pair it with another theory, that includes $\widetilde{{\rm G}} := {\rm U}(1) \times {\rm SU}(2) \times {\rm SU}(3)$ gauge group, and thus we consider the gauge group $\mathbf{G} = \widetilde{{\rm G}} \times [{\rm SU}(2) \times {\rm SU}(2)]$. But, which action should we use for this gauge group?

We do not use the Yang-Mills gauge theory, because it is metric dependent. As such, Yang-Mills action is not an appropriate choice. In the quantization of gravity, it is important that any observable computed, must be a diffeomorphism invariant. See \cite{Smolin2006-SMOTCF, Thiemann:2007zz}. The gauge theory we choose must also be diffeomorphism invariant. In other words, the field theory we choose for the gauge group $\widetilde{{\rm G}}$, should be a topological quantum field theory.

We propose using the Chern-Simons gauge theory, even though it is a 2+1 dimensional field theory. In the model of anyons, one can consider a Yang-Mills action, together with a Chern-Simons action. The Chern-Simons action will dominate the Yang-Mills action. Furthermore, the Chern-Simons theory emerge as a possible generalization of 3+1 electromagnetic ${\rm U}(1)$ gauge theory to 2+1 dimensions. See \cite{Ademola}. Another example is in the study of quantum hall fluids. One considers the movement of the electrons to be restricted to a plane. For long distances, the Chern-Simons action will dominate over a Maxwell action. See \cite{Zee}. In both models, we can see that  Chern-Simons action is a dominant action and 2+1 dimensional field theories are indeed relevant.

Witten in \cite{MR990772} showed that a Chern-Simons path integral with a ${\rm SU}(N)$ gauge group, will yield the Jones polynomial for a link, up to a physicists' level of rigor. Freed in \cite{Freed}, termed these invariants as quantum Chern-Simons invariants, to be distinguished from the classical Chern-Simons invariants. This coins the term topological quantum field theory in $\bR^3$. In the Chern-Simons gauge theory, there is no restriction on the compact gauge group used.

In \cite{CS-Lim01, CS-Lim02}, we made sense of the Chern-Simons path integral using an Abstract Wiener space formalism, and computed the path integral rigorously. Indeed, the Chern-Simons path integral, also known as Wilson Loop observable given by Equation (\ref{e.w.3}), does give us Homfly-type of polynomial invariants for a link, when the gauge group is ${\rm SU}(N)$. In the quantum Chern-Simons theory, each component knot is linked with each other, giving rise to crossings in a link diagram, each described by a $R$-matrix. When the gauge group is abelian, the quantum Chern-Simons invariants will yield the linking number of the link in $\bR^3$.

Now the Standard Model is a 4-dimensional theory. But yet we chose a 3-dimensional theory Chern-Simons theory. How is this justifiable? The authors in \cite{PhysRevLett.61.1155} discussed how a successful quantization of gravity, should yield link invariants in $\bR^3$, using a path integral. Because the quantum Chern-Simons path integral yield link invariants, it would then seem that using a Chern-Simons action would be most appropriate.

The hyperlinking number is not an invariant under time-like isotopy. As to be explained in Section \ref{s.sum}, there is no non-trivial definition of a linking number between loops in $\bR^4$. In quantum geometry, we need to replace the concept of a time-axis, with time-ordering, which also implies causality. Hence, we have to impose time-ordering between pairs of matter and geometric loops in the tangled hyperlink $\chi(\oL, \uL)$. See \cite{EH-Lim06}. Time-ordering is associated with Hamiltonian constraint in LQG, as explained in greater detail in \cite{EH-Lim07}.

This will eventually lead us to compute the path integrals, from a link diagram. Equation (\ref{e.a.4}) shows that the Wilson Loop observables are defined using linking number of a link in spatial $\bR^3$, projected down from a time-like hyperlink $\chi(\oL, \uL)$. In 3-dimensional topological theory, one of the link invariants we have is the linking number.  Therefore, it is no loss of generality then, to choose a 3-dimensional Quantum Field Theory. In that sense, we see that in both theories, linking numbers between knots, will appear in the calculations, even though the formulation in LQG is in 4-dimensional space-time.

In \cite{Witten:1988hc}, Witten talked about how the Einstein-Hilbert action will reduce to a Chern-Simons action in $2 + 1$ gravity, showing a connection between the Chern-Simons action and the Einstein-Hilbert action in 3-dimension. The final justification for using the Chern-Simons action is in computing the Einstein-Hilbert path integrals. In \cite{EH-Lim02}, we showed how an Einstein-Hilbert path integral can be written in a form of a Chern-Simons path integral, which allowed us to compute the Wilson Loop observables. With all these above reasons, the action we will choose for the gauge group $\widetilde{{\rm G}}$, is the Chern-Simons action, compatible with Einstein-Hilbert action.

In LQG, we can quantize physical observables like area and curvature of a surface, and the volume of a solid region, into operators. We are curious to know if we can use the Chern-Simons path integral to define a quantum operator for some physical observable. It would appear then that the quantum Chern-Simons theory might be contained in a bigger theory. This is where we need gravity to complete this theory. For a complete unified theory, matter must interact with space.

A theory that unites all the fundamental forces, including gravity, should itself, contain a quantum theory of gravity. If LQG is indeed a quantum theory of gravity, then we need to find a theory that can incorporate LQG as a constituent in this theory. By adding Chern-Simons and Einstein-Hilbert actions together, we can extend LQG to include quantum Chern-Simons theory, to form a complete theory, involving $\widetilde{{\rm G}} \times [{\rm SU}(2) \times {\rm SU}(2)]$ gauge group. Just like how one unifies the 3 forces by considering the direct product of gauge groups in the Standard Model Lagrangian, we will do the same by adding ${\rm G}_0:= {\rm SU}(2) \times {\rm SU}(2)$ to the gauge group $\widetilde{{\rm G}}$.

Both theories are topological quantum field theories, invariant under diffeomorphism of spatial $\bR^3$. This is also related to a diffeomorphism constraint imposed in LGQ. Refer to \cite{EH-Lim07}. The common denominator in both theories will be a link in spatial $\bR^3$, which is projected down from a time-like hyperlink in $\bR \times \bR^3$, to be defined in Definition \ref{d.tl.1}. We will define a holonomy of the gauge group $\widetilde{{\rm G}} \times {\rm G}_0$, taken over a hyperlink, describing all the fundamental forces, including gravity, and average it using a path integral expression, which we will also refer it as a Wilson Loop observable.


\section{Summary of main results}\label{s.sum}


In the 4-manifold $\bR^4 \cong \bR \times \bR^3$, $\bR$ will be referred to as the time-axis and $\bR^3$ is the spatial 3-dimensional Euclidean space. Fix a coordinate axes for $\bR \times \bR^3$, and let $\{e_a\}_{a=0}^3$ be the standard orthonormal basis for $\bR \times \bR^3$, with $\{e_i\}_{i=1}^3$ being the standard basis in $\bR^3$. And $\Sigma_i$ is the plane in $\bR^3$, containing the origin, whose normal is given by $e_i$.

Pertaining to this standard basis $\{e_a\}_{a=0}^3$, let $\vec{x} =(x_0, x_1, x_2, x_3)$ be the standard coordinates on $\bR^4$,  whereby $x_0$ will be referred to as time. Therefore, $\Sigma_1$ is the $x_2-x_3$ plane, $\Sigma_2$ is the $x_3-x_1$ plane and finally $\Sigma_3$ is the $x_1-x_2$ plane. Note that $\bR \times \Sigma_i \cong \bR^3$ is a 3-dimensional subspace in $\bR^3$ and let $\pi_i: \bR^4 \rightarrow \bR \times \Sigma_i$ denote this projection. Let $\pi_0: \bR \times \bR^3 \rightarrow \bR^3$ denote a projection.

Suppose $E$ is a trivial bundle over $\mathbb{R}^4 \equiv \bR \times \bR^3$, with any compact structure group ${\rm G}:= {\rm G}_1 \times \cdots \times {\rm G}_{\bar{m}}$, each Lie group ${\rm G}_a$ is compact. We will assume that ${\rm G}_1 = {\rm U}(1)$ and ${\rm G}_a$ is non-abelian, $a \neq 1$. This generalizes the gauge group ${\rm U}(1) \times {\rm SU}(2) \times {\rm SU}(3)$ in the Standard model. Let $\mathfrak{g}$ and $\mathfrak{g}_a$ denote the Lie Algebra of ${\rm G}$ and ${\rm G}_a$ respectively.

Let $\mathfrak{g}_0$ be the Lie algebra of ${\rm G}_0 = {\rm SU}(2) \times {\rm SU}(2)$. For each $a=0, 1, \cdots, \bar{m}$, let $q_a$ be a scalar quantity, called the charge pertaining to the gauge group ${\rm G}_a$, with corresponding Lie algebra $\mathfrak{g}_a$.

For a finite set of non-intersecting simple closed curves in $\bR^3$ or in $\bR \times \Sigma_i$, we will refer to it as a link. If it has only one component, then this link will be referred to as a knot. A simple closed curve in $\bR^4$ will be referred to as a loop. A finite set of non-intersecting loops in $\bR^4$ will be referred to as a hyperlink in this article. If we assign an orientation to each component loop, then the hyperlink is said to be oriented.

Two loops are linked together if it is impossible to translate one loop by an arbitrary distance from the other loop without the two objects actually crossing one another. See \cite{greensite2011introduction}. In $\bR^4$, one can topologically deform and `unlink' the loops, without crossing. As such, if we use ambient isotopy as an equivalence relation, then an equivalence class of a set of loops in $\bR^4$ will be a set of `unlinked' simple closed curves. Hence we will only consider a special equivalence class of hyperlinks.

\begin{defn}(Time-like hyperlink)\label{d.tl.1}\\
Let $L$ be a hyperlink. We say it is a time-like hyperlink if given any 2 distinct points $\vec{p}\equiv (x_0, x_1, x_2, x_3), \vec{q}\equiv (y_0, y_1, y_2, y_3) \in L$, $\vec{p} \neq \vec{q}$,
\begin{enumerate}
  \item (T1) $\sum_{i=1}^3(x_i - y_i)^2 > 0$;
  \item (T2) if there exist $i, j$, $i \neq j$ such that $x_i = y_i$ and $x_j = y_j$, then $x_0 - y_0 \neq 0$.
\end{enumerate}
\end{defn}

All our hyperlinks are assumed to be time-like, as given in Definition \ref{d.tl.1}. By its definition, the projection of a time-like hyperlink $L$ using $\pi_a$, $a=0, 1, 2, 3$, will yield a link.

Suppose we have two oriented time-like hyperlinks $\oL = \{\ol^1, \ldots, \ol^{\on}\}$ and $\uL = \{\ul^1, \ldots, \ul^{\un}\}$, which are tangled together, to form an oriented time-like hyperlink $\chi(\oL, \uL)$. The former (latter) will be referred to as a matter (geometric) hyperlink. The oriented time-like hyperlink $\oL$ is assumed to be a framed hyperlink, i.e. this means we define a frame on each $\pi_0(\ol^u)$, which is a normal vector field $v^u =(v_1^u, v_2^u, v_3^u) \in \bR^3$ along $\pi_0(\ol^u)$ that is nowhere tangent to it. Hence, $\pi_0(\oL) \equiv \{\pi_0(\ol^1), \cdots, \pi_0(\ol^{\on})\}$ is a framed link (or ribbon), with nodes assigned to each component knot $\pi_0(\ol^u)$. See \cite{EH-Lim04} for more details. We do not assign any framing to the geometric hyperlink $\uL$.

For each component matter loop $l^u \subset \oL$, we colour it with an irreducible representation $\varrho_u: \mathfrak{g} \rightarrow {\rm End}(W_u)$, and  $\rho_u\equiv (\rho_u^+, \rho_u^-): \mathfrak{su}(2) \times \mathfrak{su}(2) \rightarrow {\rm End}(V_u^+) \times {\rm End}(V_u^-)$. Each representation $\rho_u^\pm$ is assumed to be irreducible, hence the dimension of $V_u^\pm$ is $2j_u^\pm + 1$, whereby $j_u^\pm$ is an integer or half-integer. Refer to Remark \ref{r.l.3}. Define $|j| := \sum_{u=1}^{\on} |j_u^+| + |j_u^-|$. Each $j_u^\pm$ will define a quadratic Casimir operator, proportional to identity. This proportionality constant is given by $j_u^\pm(j_u^\pm+1)$, and will be interpreted as energy. The higher the value of $j_u^\pm$, the higher is the energy level.

We do not colour the geometric component loop with a representation. We will term $\chi(\oL, \uL)$ as a coloured time-like hyperlink. For each component matter loop $\ol^u$, we time-ordered it with every component geometric matter loop $\ul^v$, and consider $\chi(\oL, \uL)$, up to time-like isotopy and time-ordered equivalence relation, as defined in \cite{EH-Lim06}. Time-like isotopy and time-ordered equivalence relation are associated with the diffeomorphism constraint and Hamiltonian constraint respectively. Refer to \cite{EH-Lim07}.

In Section \ref{s.uni}, we will explain how we can define quantum Chern-Simons invariants on a time-like hyperlink $\oL \subset \bR^4$. Suppose we have a sequence of coloured time-like matter hyperlinks $\{ \oL^n \}_{n=1}^\infty$, all assumed to be inequivalent time-like hyperlinks.  Further assume that we can distinguish them using quantum Chern-Simons invariants. Refer to Item \ref{r.l.1c} in Remark \ref{r.l.1}.

For each $k \in \mathbb{N}$, $\oL^k = \{\ol^{k, 1}, \cdots, \ol^{k, \on_k}\}$, each matter loop $\ol^{k, u}$ is coloured by some representation $(\varrho_u^k, \rho_u^k)$. Suppose for each $\oL^k$, there is a corresponding time-like geometric hyperlink $\uL^k = \{\ul^{k, 1}, \cdots, \ul^{k, \un_k}\}$, $\on_k, \un_k \in \mathbb{N}$. These two hyperlinks are tangled together to form a coloured time-like hyperlink $\chi(\oL^k, \uL^k)$, for each $k \in \mathbb{N}$. Construct the sequence $\{ \chi(\oL^k, \uL^k) \}_{k=1}^\infty$.

Using a non-trivial representation $\{\rho_u\}_{u=1}^{\on}$, a quantum Einstein-Hilbert invariant given by Equation (\ref{e.weh.1}), only takes into account the hyperlinking number between each matter loop $\ol^u$, $u=1, \ldots, \on$, and $\uL$. It can only show that each component matter loop is `linked' with a geometric hyperlink; it may not be able to differentiate $\oL^n$ from $\oL^m$, $m \neq n$, in the sequence $\{\oL^n \}_{n=1}^\infty$. Since the quantum Einstein-Hilbert invariant does not contain any Homfly-type of polynomial invariant, we will then ask if it is possible to include the Chern-Simons action and write down a unified path integral expression, that includes a quantum Chern-Simons invariant.

The answer is yes. In Theorem \ref{t.u.1}, we will compute the average of the holonomy of a $\mathfrak{g} \times \mathfrak{g}_0$-valued connection, over a hyperlink $\chi(\oL, \uL)$. This average is given by a path integral Expression \ref{ex.toe.1}, using both the Chern-Simons and Einstein-Hilbert action. This average of the holonomy over $\chi(\oL, \uL)$, termed as Wilson Loop observable, can be computed and is expressed as a state model.  Indeed, it is a link invariant containing quantum Chern-Simons invariants, which satisfy the Homfly type of skein relations, under certain quantization conditions. Refer also to Remark \ref{r.l.1}.

When we consider the fundamental representation of ${\rm SU}(N)$, $N > 1$, the charge $q_a$ being an integer or half-integer, is both necessary and sufficient for the quantum Chern-Simon invariants to satisfy the Yang-Baxter equation. Hence, it cannot be made small arbitrary. When $q_a^2$ is an even integer, we see that the quantum Chern-Simons invariant just gives us the linking number. Refer to Remark \ref{r.r.1}. On the other hand, there is no restriction on $q_0$ and $q_1$, whereby $q_0$ and $q_1$ are the charges corresponding to ${\rm G}_0$ and ${\rm U}(1)$ gauge groups respectively. In the case when the gauge group is abelian, the quantum Chern-Simons invariants essentially give us the linking number. Thus, a non-abelian quantum Chern-Simons invariant, is a stronger `link invariant' than one using an abelian gauge group.

When the fundamental forces are represented by non-trivial representations \\
$\{\varrho_u\}_{u=1}^{\on}$, the Wilson Loop observable given by Equation (\ref{e.z.2}) in Theorem \ref{t.u.1}, contains both the quantum Chern-Simons invariants and the Einstein-Hilbert invariants.  The quantum Chern-Simons invariants capture how matter interact with itself; the quantum Einstein-Hilbert invariants capture how matter interact with space. The former, is able to distinguish the hyperlinks in the sequence $\{ \chi(\oL^k, \uL^k) \}_{k=1}^\infty$ by our assumption.

The quantum Chern-Simons invariants are different from the Einstein-Hilbert invariants. After imposing time-ordering, the quantum Einstein-Hilbert invariants just give us the linking number, between a projected matter loop and a projected geometric hyperlink, in spatial $\bR^3$. See Remark \ref{r.l.2}. Under our earlier assumption, each inequivalent time-like hyperlink in the sequence $\{ \oL^n \}_{n=1}^\infty$ is distinguishable by a quantum Chern-Simons invariant. If we boldly assume that nature does not require two invariants to distinguish two inequivalent time-like hyperlinks, then $q_0$ has to be made small, compared to $q_a$, $a = 1, \cdots, \bar{m}$. This might give a plausible reason why gravity is weak. This is also analogous to how in the model of anyons or quantum hall fluids, the Chern-Simons action dominates over the Yang-Mills action.

We can generalize Theorem \ref{t.u.1} in Section \ref{s.uni}. Consider an oriented triple \\
$\{S, \emptyset, \chi(\oL, \uL)\}$ as described in Theorem 3 of \cite{EH-Lim06}, up to an equivalence relation. In \cite{EH-Lim05}, we quantized spin curvature, taken over a closed surface $S$ in $\bR^4$, and showed that it can be computed using the linking number between $S$ and $\uL$. This is an invariant up to ambient isotopy. We will prove that a path integral given by Expression \ref{ex.toe.3}, which generalizes Expression \ref{ex.toe.1}, shows that the Wilson Loop observable given by Equation (\ref{e.z.2}), is an eigenfunctional for a quantized spin curvature $\hat{F}_S$. This is the content of Theorem \ref{t.u.3}. This eigenvalue, can be calculated from the linking number between $S$ and $\uL$, given by Equation (\ref{e.toe.4}). It requires the global topology of both $\chi(\oL, \uL)$ and the surface $S$, hence it is not a local invariant.

The Principle of Equivalence says that when gravitational effects are observed, it is impossible, by any experiment, to determine the type of gravitational field responsible for it. When the quantized curvature of a surface $S$ is non-zero, we know the presence of a geometric hyperlink, which represents gravitons. See \cite{EH-Lim07}.

Recall we had a sequence $\{\chi(\oL^n, \uL^n)\}_{n=1}^\infty$. Choose a closed surface $S$ and decide how each geometric hyperlink $\uL^n$ is `linked' with $S$. Construct a sequence of oriented triple $\{S, \emptyset, \chi(\oL^n, \uL^n)\}_{n=1}^\infty$, for the fixed closed surface $S$. Because of the Principle of Equivalence, the eigenvalues of the quantized curvature must be degenerate. In other words, knowing the linking number between $S$ and a geometric hyperlink, is insufficient to tell us which hyperlink $\chi(\oL^n,\uL^n)$ in the sequence corresponds to the said eigenvalue. For details, see \cite{EH-Lim07}. But, looking at Equation (\ref{e.toe.4}), we can still determine the time-like hyperlink $\chi(\oL, \uL)$, solely from the quantum Chern-Simons invariants, when the representation of $\mathfrak{g}$ is non-trivial. This does not contradict the Equivalence Principle.

We will move on to quantize area of a compact surface $S$, and volume of a compact solid region $R$ in $\bR^3$, details to be given in Section \ref{s.av}. When we quantized area and volume in LQG, we can consider a quanta of area or volume in quantum geometry. This happens when we go to Planck's scale of distance. In fact, to define the area and volume path integrals, it is necessary to partition the surface $S$ and region $R$. This further shows that the LQG is a local theory. Refer to \cite{EH-Lim03, EH-Lim04, EH-Lim07} for details.

\begin{rem}
\begin{enumerate}
  \item In the quantization of volume, we need to consider nodes on $\pi_0(\oL)$, which are 0-dimensional. This shows that LQG is a local theory, in the sense defined in \cite{Freed}.
  \item In the absence of a metric on $\bR^4$, it does not make sense to talk about `closeness' or `local'. When we say it is local, we mean that we are able to compute a quanta of area or volume.
\end{enumerate}
\end{rem}

Consider an oriented triple $\{S, R, \chi(\oL, \uL)\}$, up to an equivalence relation, as described in Theorem 3 of \cite{EH-Lim06}. We will show that using path integral Expressions \ref{ex.as.1}, \ref{ex.as.2} and \ref{ex.vs.1} respectively, the Wilson Loop observable given by Equation (\ref{e.z.2}), is not an eigenfunctional of the area operator $A_S$, or volume operator $V_R$. See Theorems \ref{t.u.4} and \ref{t.u.2}.

The Wilson Loop observable given by Equation (\ref{e.z.2}), will be an eigenfunctional, provided we choose the representation of $\mathfrak{g}$ to be trivial.  The implication of this will be that we will no longer be able to identify the matter hyperlink $\oL$ using the quantum Chern-Simons invariants. The matter hyperlink representing particles, no longer transform under the gauge group, which happens when area or volume is quantized.

When the representations $\{\ol^u,\varrho_u\}_{u=1}^{\on}$ for a coloured time-like hyperlink $\oL$ becomes trivial, we see that Equation (\ref{e.z.2}) reduces to $Z(q_0/3; \{\ol^u, \rho_u\}_{u=1}^{\on}, \uL)^3$. Since we assume $q_0$ to be small, using Taylor's expansion, we see that the quantum Einstein-Hilbert invariant in Equation (\ref{e.n.1}) is \beq Z(q_0; \{\ol^u, \rho_u\}_{u=1}^{\on}, \uL) = \sum_{u=1}^{\on} (2j_u^+ + 1) + (2j_u^- + 1) + O\left( q_0^2 |j|^3 \right). \nonumber \eeq When $|j|$ is small, i.e. at low energies, it is very hard to detect the hyperlinking number. That explains why it is so difficult to detect quantum gravity.

Now consider the quantized curvature operator, given by Equation (\ref{e.ehc.2}). Using Taylor expansion, it is written as \beq {\rm lk}(S, \uL) \left[\sum_{u=1}^{\on} (2j_u^+ + 1) + (2j_u^- + 1) + O\left( q_0^2 |j|^3 \right) \right]\otimes \mathcal{E}. \nonumber \eeq Due to Principle of Equivalence, knowing the linking number ${\rm lk}(S, \uL)$ is insufficient to determine the corresponding hyperlink $\chi(\oL, \uL)$. To identify the hyperlink, it is necessary to determine the hyperlinking number, given by the coefficient of $q_0^2|j|^3$ term.

When $q_0|j|^2$ is small, by doing a Taylor expansion, the area path integral given by Equations (\ref{e.a.3a})-(\ref{e.a.3b}), are of the order $q_0|j|^2$. Thus it is easier to detect a quanta of area, rather than the hyperlinking number, which appears as coefficient of $q_0^3|j|^5$. The term corresponding to $q_0|j|^2$ becomes significant when $|j|$ is large, or at high energies.

In a similar manner, the volume path integral given by Equation (\ref{e.v.1}), is of the order $q_0^2|j|^3$. As such, it is easier to detect a quanta of area, then a quanta of volume. The hyperlinking number, will be difficult to detect as it is the coefficient of term $q_0^3|j|^5$ in the area path integral expression.

Recall we had a sequence $\{\chi(\oL^n, \uL^n)\}_{n=1}^\infty$. Fix a compact surface $S \subset \bR^3$ and a solid region $R \subset \bR^3$, and decide how each $\pi_0(\oL^n)$ `pierce' the surface $S$ and how many of the nodes in $\pi_0(\oL^n)$ lie inside $R$, as described in \cite{EH-Lim06}. Construct a sequence of oriented triple $\mathbb{T} := \{S, R, \chi(\oL^n, \uL^n)\}_{n=1}^\infty$.

When one obtains a non-zero value from the area or volume operator, it is enough to conclude that matter hyperlinks are present. Moreover, the authors in \cite{Rovelli:1995ac} claimed that volume and area operators are sufficient to distinguish all the time-like hyperlinks in the sequence from each other. This means, in the construction of the sequence in the preceding paragraph, we require that the area and volume eigenvalues, are non-degenerate. Thus, it is possible to identify the corresponding hyperlink $\chi(\oL^n, \uL^n)$ from the sequence, using the eigenvalues of the quantized area and volume operators, which also give us some information on the representations $\{\rho_u^n\}_{u=1}^{\on}$. Refer to \cite{EH-Lim07}.

To summarize, it is only at short distances, that gravity plays an important role. The preference for certain type of quantum invariants might explain why gravity is very weak compared to the other forces, i.e. $q_0$ is small. To be able to identify the corresponding time-like hyperlink in the sequence, it is necessary to compute the coefficients of $q_0|j|^2$ or $q_0^2|j|^3$ terms, possible only when $j$ is large or at high energies.

\section{Chern-Simons action}\label{s.cs}

Ever since Witten published his paper \cite{MR990772}, several authors have attempted to make rigorous sense of the Chern-Simons path integral and prove that it defines a Homfly type of polynomial for a link. Here, we will summarise the results obtained from our previous work on defining Chern-Simons path integrals.

Since ${\rm G}$ is compact, we can assume $\mathfrak{g}$ is a Lie subalgebra of the Lie algebra $\mathfrak{u}(N)$ of ${\rm U}(N)$, for some $N \in \mathbb{N}$. Define an inner product $\langle A, B \rangle := -\Tr[AB]$, $\Tr$ is the usual matrix trace. Fix an orthonormal basis $\{E_a^\alpha: a=1, \cdots, \bar{m}, \alpha=1, \cdots, N_a\} \subset \mathfrak{g}$ throughout this article.

Let $x=(x_1, x_2, x_3)$ be the standard coordinates on $\mathbb{R}^3$, pertaining to some orthonormal basis $\{e_a\}_{a=1}^3$ fixed in $\bR^3$. In \cite{CS-Lim01, CS-Lim02}, we described how to quantize the Chern-Simons gauge theory for both abelian and non-abelian gauge group $G$ respectively.

The Chern-Simons action is given by \beq S_{{\rm CS}}(A) = \frac{\varsigma}{4\pi} \int_{\mathbb{R}^3} \Tr \left[A \wedge dA + \frac{2}{3}A \wedge A \wedge A\right],\ \varsigma \neq 0, \label{e.cs.2} \eeq $A$ is a $\mathfrak{g}$-valued connection on $\mathbb{R}^3$, which decays to 0 fast enough for it to be integrable on $\mathbb{R}^3$.

Now, we want to make sense of an expression \beq Z_{{\rm CS}} := \int_{ A \in \mathscr{A}} e^{i S_{\rm CS}(A)} DA, \label{e.cs.5} \eeq whereby $DA$ is a non-existent Lebesgue type of measure, $\mathscr{A}$ is the space of connections, modulo gauge transformations, and $i = \sqrt{-1}$.

\begin{notation}
Using standard coordinates $\vec{x} = (x_0, x_1, x_2, x_3) = (x_0, x)$ on $\bR^4$, let $\Lambda^1(\bR^4)$ be the vector space spanned by $\{dx_0, dx_1, dx_2, dx_3\}$, with subspace $\Lambda^1(\bR^3)$ spanned by $\{dx_1, dx_2, dx_3\}$.

For $\kappa > 0$, let $\overline{\mathcal{S}}_\kappa(\bR^3) \subset L^2(\bR^3)$ be a Schwartz space, which consists of functions of the form $f = p \cdot \sqrt{\psi_\kappa}$, whereby $p$ is a polynomial on $\bR^3$, and $\psi_\kappa(x) = \frac{\kappa^3}{(2\pi)^{3/2}}e^{-\kappa^2|x|^2/2}$ is the Gaussian function with variance $1/\kappa^2$.
\end{notation}

Fix $a \in \{1, \cdots, \bar{m}\}$. Every $\mathfrak{g}_a$-valued connection over $\mathbb{R}^3$ can be gauge transformed into the form (sum over repeated index $\alpha=1, \cdots, N_a$, for each $a$) \beq A_a = A_\alpha^{a,1} \otimes dx_1\otimes E_a^\alpha + A_\alpha^{a,2} \otimes dx_2\otimes E_a^\alpha, \nonumber \eeq where $A_\alpha^{a,1}, A_\alpha^{a,2}$ are real-valued functions on $\mathbb{R}^3$ and \beq A_\alpha^{a,2}(x_0,x_1,0) = 0,\ A_\alpha^{a,1}(x_0, 0,0) = 0. \label{e.c.1} \eeq This is called Axial gauge fixing in \cite{CS-Lim01}. Assume $A_\alpha^{a,1}$ and $A_\alpha^{a,2}$ are smooth and decay fast enough for the integral to be defined. We will also drop the restrictions in Equation (\ref{e.c.1}). Hence, we will write $A_a = \sum_{p=1}^2 A_\alpha^{a,p} \otimes dx_p \otimes E_a^\alpha$, summing over index $\alpha$. With this gauge the cubic term in the action drops out.

Define for $p=1,2$,
\begin{align}
K_p^{\kappa} :=& \left\{ \sum_{a=1}^{\bar{m}} A_\alpha^{a,p} \otimes dx_p \otimes E_a^\alpha:\ A_\alpha^{a,p} \in \overline{\mathcal{S}}_\kappa(\bR^3) \right\}.\label{e.o.2}
\end{align}
Thus, our connection $A= \sum_{a=1}^{\bar{m}}A_a$ is in $K_1^{\kappa} \oplus K_2^{\kappa} \subset \overline{\mathcal{S}}_\kappa(\bR^3) \otimes \Lambda^1(\bR^3) \otimes \mathfrak{g}$. Note its dependence on the parameter $\kappa$.

The Chern-Simons action will be used to describe the forces, for example the weak, strong and electromagnetic forces. The bosons mediating the force $F_a$, will be described by the generators of the Lie algebra $\mathfrak{g}_a$, represented as an endomorphism on some vector space. Refer to \cite{GUT}.

\begin{defn}\label{d.t.1}(Time ordering operator)\\
Suppose a matrix $A(s)$ is indexed by $s \in \bR$. For any permutation $\sigma \in S_r$, \beq \mathcal{T}(A(s_{\sigma(1)})\cdots A(s_{\sigma(r)})) = A(s_1)\cdots A(s_r),\ s_1 > s_2 > \ldots > s_r. \nonumber \eeq
\end{defn}

Suppose we have a link $L = \{ l^u:\ u=1, \ldots, \on\} \subset \bR^3$. For each knot $l^u$, we colour it with an irreducible representation $\varrho_u: \mathfrak{g} \rightarrow {\rm End}(W_u)$.

Of great interest is the evaluation of
\begin{align}
Z(q;& \{l^u, \varrho_u\}_{u=1}^{\on}) \nonumber \\
&:= \lim_{\kappa \rightarrow \infty}\frac{1}{Z_{{\rm CS}}^\kappa}\int_{A= \sum_{a=1}^{\bar{m}}A_a \in K_1^\kappa \oplus K_2^\kappa} W(q; \{l^u, \varrho_u\}_{u=1}^{\on})(A) e^{iS_{\rm CS}(A)} D[A], \label{e.w.2}
\end{align}
where (sum over $p=1,2$ and $\alpha=1, \cdots, N_a$)
\beq
W(q; \{l^u, \varrho_u\}_{u=1}^{\on})(A) := \prod_{u=1}^{\on}\Tr_{\varrho_u}  \mathcal{T} \exp\left[ \sum_{a=1}^{\bar{m}}q_a\int_{l^u} A^{a,p}_{\alpha} \otimes dx_p\otimes E_a^{\alpha}  \right] \label{e.b.1} \eeq and \beq Z_{{\rm CS}}^\kappa =
\int_{A \in K_1^\kappa \oplus K_2^\kappa}  e^{iS_{\rm CS}(A)} DA, \nonumber \eeq for a non-existent Lebesgue type of measure $D[A]$.

Here, $\Tr_{\varrho_u}$ is the matrix trace in the representation $\varrho_u$ for $\mathfrak{g}\equiv \mathfrak{g}_1 \times \mathfrak{g}_2 \times \cdots \times \mathfrak{g}_{\bar{m}}$ and $\mathcal{T}$ is the time ordering operator, defined in Definition \ref{d.t.1}. The integral in Equation (\ref{e.w.2}) will be known as the Wilson Loop observable (associated to the link $L$). The vector $q = (q_1, \cdots, q_{\bar{m}})$ will be called the charge of the link. Note that
$\Tr_{\varrho_u}  \mathcal{T} \exp\left[ \sum_{a=1}^{\bar{m}}q_a\int_{l^u} A_a  \right]$ is the holonomy of $\sum_{a=1}^{\bar{m}}q_a\varrho_u(A_a)$, along the knot $l^u$.

In \cite{CS-Lim01}, we showed how to make sense of the Chern-Simons path integral in Equation (\ref{e.w.2}) using Wiener measure, for each $\kappa$. Its limit as $\kappa \rightarrow \infty$, can be computed from a link diagram in $\Sigma_3$, proved in \cite{CS-Lim02}. To explain the result, we need to digress a bit and discuss link diagrams. On each link diagram, we have crossings, which we also refer to as double points in \cite{CS-Lim02}. We can also define a directed graph from this link diagram.

\section{Link diagrams}\label{s.ld}

\begin{defn}\label{d.c.1}(Edges)\\
Fix $i=1, 2, 3$. An oriented link $L = \{l^u\}_{u=1}^{\on} \subset \bR^3$ is projected on $\Sigma_i \cong \mathbb{R}^2$ forming a planar directed graph as follows. Let $C^u: [0,1] \rightarrow \bR^3$ be a parametrization of $l^u$, consistent with the orientation.
\begin{enumerate}
  \item The vertex set $V(\Sigma_i; L)$ will be identified with the set of double points, more commonly known as crossings, on a link diagram in $\Sigma_i$, denoted as $\pd(\Sigma_i; L)$. The set of edges $E(\Sigma_i; L)$ is simply the set of lines in the planar diagram of $L$ joining each vertex. Each edge $e: [\epsilon_1,\epsilon_2] \rightarrow \mathbb{R}^2$, $0 \leq \epsilon_1 < \epsilon_2 \leq 1$. The end points $e(\epsilon_1), e(\epsilon_2)$ will be a vertex in $V(\Sigma_i; L)$. Each vertex has 4 edges incident onto it.
  \item Fix a $u$. For each knot $l^u$ and $l^v$, $v \neq u$, we let $\pd(\Sigma_i; l^u, l^v)$ denote all the crossings on the link diagram in the plane $\Sigma_i$, each such crossing is formed from projecting an arc in $l^u \subset L$ and another arc in $l^v \subset L$ onto $\Sigma_i$. For each crossing $p$ in $\pd(\Sigma_i; l^u, l^v)$, $u \neq v$, let $V_u(\Sigma_i; l^v)$ be the set of vertices corresponding to such points $p$.
  \item Suppose $p \in \pd(\Sigma_i; l^u)$, a set of crossings formed using only arcs in $l^u$. Let $V_u(\Sigma_i; l^u)$ be the set of vertices corresponding to such $p$'s on a planar diagram of $l^u$.
  \item\label{d.c.1.7}  Define $V(\Sigma_i; l^u) = \bigcup_{v=1}^{\on} V_u(\Sigma_i; l^v)$, which defines the vertex set of the graph $l^u$. The set of edges $E(\Sigma_i; l^u)$, is a subset of $E(\Sigma_i; L)$, joining only vertices in $V(\Sigma_i; l^u)$. Note that $E(\Sigma_i; L) = \bigcup_{v=1}^{\on} E(\Sigma_i; l^v)$.
  \item Suppose $e$ and $\hat{e}$ belong to $E(\Sigma_i;l^u)$. We say that an edge $e: [\epsilon_1,\epsilon_2] \rightarrow \mathbb{R}^2$ precedes another edge $\hat{e}:[\hat{\epsilon}_1,\hat{\epsilon}_2] \rightarrow \mathbb{R}^2$ if $\epsilon_2 \leq \hat{\epsilon}_1$.
  \item Each crossing $p \in \pd(\Sigma_i; L)$ will be denoted by 4 edges, labeled by \newline
  $(e^+(p), e^-(p), \bar{e}^+(p), \bar{e}^-(p))$, whereby $e^+$ and $e^-$ are edges belonging to $E(\Sigma_i; l^u)$ with the bigger index $u$ and $e^-(\bar{e}^-)$ is the edge that precedes $e^+(\bar{e}^+)$ at the vertex $p$. When all 4 edges belong to the same curve, then $\bar{e}^+$ and $\bar{e}^-$ are the edges that precede $e^+$ and $e^-$ respectively.
  \item\label{i.f.1} Now suppose we define a frame on $L$ and project the framed oriented link onto $\Sigma_i \cong \mathbb{R}^2$, as described in \cite{CS-Lim02}. Using this frame $v^u$, which is a nowhere tangent vector field along $l^u$, we can define a `duplicate' copy of $l^u$, denoted by $l^{u,\epsilon}:= l^u + \epsilon v^u$, $\epsilon > 0$ small.
    Thus, $l^u$ and $l^{u,\epsilon}$, will together define a ribbon, which gives rise to half-twists and hence nodes in $l^u$.

    For each component framed knot $l^u$, the crossings in the planar diagram formed from $l^u$ and its copy $l^{u, \epsilon}$, will define a set of twisted double points, more commonly known as half twists, denoted as ${\rm TDP}(\Sigma_i;l^u)$. Such a set of half twists will define a set of vertices as in the case of an oriented link. A half twist $q$ will then be represented by a vertex, also referred to as node, with only 2 edges incident onto it, labeled $(e^+(q), e^-(q))$. Thus, a full twist, given by 2 consecutive half twists, twisted in the same direction, will be represented in the planar graph of the curve $l^u$ by 2 vertices, joined together by a common edge. Denote ${\rm TDP}(\Sigma_i;L) = \bigcup_{u=1}^{\on}{\rm TDP}(\Sigma_i;l^u)$.
\end{enumerate}
\end{defn}

For each crossing $p$, we can define its orientation and height, hence define its algebraic number $\varepsilon(p)$. A twisted double point can also be assigned an algebraic number. Refer to \cite{CS-Lim02}. Given two distinct oriented knots $l^u$ and $l^v$, its linking number is given by \beq {\rm lk}(l^u, l^v) = \sum_{p \in \pd(\Sigma_k; l^u, l^v)} \varepsilon(p), \nonumber \eeq and is independent of any $k=1, 2, 3$.

\begin{rem}
Refer to Section \ref{s.sum}. Given a time-like hyperlink $\chi(\oL, \uL) \subset \bR^4$, we can project it onto $\bR^3$ using $\pi_0$ to form a link. Using ambient isotopy in $\bR^3$ if necessary, both $\pi_0(\oL)$ and $\pi_0(\uL)$ can be projected on all three planes $\Sigma_k$, $k=1,2,3$, simultaneously to form respective link diagrams $D_k \subset \Sigma_k$.

\end{rem}

\begin{notation}\label{n.q.4}
Given a time-like hyperlink $\chi(\oL, \uL)$, project it into $\bR^3$, to form a link. Further project $\pi_0(\chi(\oL, \uL))$ onto $\Sigma_i$ to form link diagram $D_i$ on $\Sigma_i$ respectively, hence each defining a directed graph, as described in Definition \ref{d.c.1}. For each $k=1,2,3$, let $\pd(\Sigma_k; \pi_0(\oL))$ denote the subset of crossings in $D_k$, each such crossing formed using arcs strictly in $\pi_0(\oL)$.

Fix a $u = 1, \cdots, \on$. From $D_k$, $k=1, 2, 3$, define
\begin{align*}
\pd(\Sigma_k; \pi_0(\ol^u), \pi_0(\uL)) &:= \bigcup_{v=1}^{\un}\pd(\Sigma_k; \pi_0(\ol^u), \pi_0(\ul^v)).
\end{align*}
Each vertex corresponding to a crossing in $\pd(\Sigma_k; \pi_0(\ol^u), \pi_0(\uL))$, lies in the directed subgraph in $\Sigma_k$, defined from the link diagram containing only $\pi_0(\ol^u)$ as described in Item \ref{d.c.1.7}. Similar to half-twists as described in Item \ref{i.f.1}, each such vertex $\tilde{q}$, is a vertex of valence 2, and will be labelled with 2 edges $(e^+(\tilde{q}), e^-(\tilde{q}))$ incident to it, in the respective subgraph.
\end{notation}

\begin{defn}\label{d.c.2}
Consider a time-like hyperlink $L \subset \bR \times \bR^3$, which we project it down onto $\bR^3$ using $\pi_0$ to form a link $\pi_0(L)$. This link can be projected down onto $\Sigma_k$ to form a link diagram $D_k$.

For a crossing $p \in D_k$, $k=1,2,3$, we defined a time lag of $p$ as was done in \cite{EH-Lim06}, which we will denote by $\tau(p) \in \{-1, +1\}$. Note that by the definition of a time-like hyperlink, each crossing has a well-defined time-lag. Hence define $\sigma(p) := \tau(p) \varepsilon(p)$.
\end{defn}

For each $\ol^u$, $u=1, \cdots, \on$, define the hyperlinking number between $\ol^u$ with $\uL$ as
\begin{align}
{\rm sk}(\ol^u, \uL):=&
\sum_{k=1}^3\sum_{p \in \pd(\Sigma_k; \pi_0(\ol^u), \pi_0(\uL))}\sigma(p). \label{e.sk.1}
\end{align}

\begin{rem}\label{r.l.4}
Note that the hyperlinking number, is not a time-like isotopic invariant. We need to time-order each pair of component matter loop and component geometric loop. See \cite{EH-Lim06}. After time-ordering, the hyperlinking number between $\ol^u$ and $\ul^v$ will be ${\rm sk}(\ol^u, \ul^v) = \pm 3\ {\rm lk}(\pi_0(\ol^u), \pi_0(\ul^v))$, the $\pm$ sign depends on how the two loops $\ol^u$ and $\ul^v$ are time-ordered. The factor 3 is due to projecting $\pi_0(\ol^u)$ and $\pi_0(\ul^v)$ simultaneously on $\Sigma_i$, $i=1,2,3$.
\end{rem}

\begin{defn}\label{d.c.3}
Define $[F_{ac}]_{ij} := \delta_{ia}\delta_{jc}$, $\delta_{kl} = 1$ if $k=l$, 0 otherwise. Let $N$, $\bar{N}$ be positive integers. For $A \in M_N(\mathbb{C}) \otimes M_{\bar{N}}(\mathbb{C})$, the components are given by $[A]^{ab}_{cd}$ with respect to the basis $\{F_{ac}\otimes F_{bd}\}_{a,b,c,d}$ of $M_N(\mathbb{C}) \otimes M_{\bar{N}}(\mathbb{C})$.

In other words, given $A \otimes B \in M_N(\mathbb{C}) \otimes M_{\bar{N}}(\mathbb{C})$, the components $[A \otimes B]^{ab}_{cd} = A^a_c \otimes B^b_d \equiv A_{ac} \otimes B_{bd}$. The indices $a, b$ track the rows, the lower indices $c, d$ track the columns.
\end{defn}

\begin{notation}\label{n.q.2}
For $u=1, \ldots, \on$, let $n_u$ be the dimension of each representation $\varrho_u: \mathfrak{g} \rightarrow {\rm End}(\bC^{n_u})$.

Given a link diagram of a link $L \subset \bR^3$ in $\Sigma_i$, a crossing $p \in \pd(\Sigma_i;L)$ is labelled with 4 edges; a half-twist $p \in  {\rm TDP}(\Sigma_i;L)$ is labelled with 2 edges.

If a crossing $p \equiv (e^+,e^-, \bar{e}^+, \bar{e}^-)$, with $\{e^+,e^-\}\subseteq E(\Sigma_i; \pi_0(\ol^u))$ and $\{\bar{e}^+,\bar{e}^-\}\subseteq E(\Sigma_i; \pi_0(\ol^v))$, then define a $R$-matrix \beq R(p) :=  \exp \left[-\varepsilon(p)\sum_{a=1}^{\bar{m}}\pi iq_a^2 \sum_{\alpha=1}^{N_a} \varrho_u(E_a^\alpha) \otimes \varrho_v(E_a^{\alpha})  \right]\in M_{n_u}(\mathbb{C}) \otimes M_{n_v}(\mathbb{C}) .\nonumber \eeq

For $A,B \in M_n(\mathbb{C})$, $A \cdot B$ means take the usual matrix multiplication. If $p \in  {\rm TDP}(\Sigma_i;\pi_0(\oL))$, with $p \equiv (e^+, e^-) \subseteq E(\Sigma_i;\pi_0(\ol^u))$, then define \beq T(p) := \exp \left[-\varepsilon(p)\sum_{a=1}^{\bar{m}}\frac{\pi iq_a^2}{2} \sum_{\alpha =1}^{N_a} \varrho_u(E_a^\alpha) \cdot \varrho_u(E_a^{\alpha})  \right] \in M_{n_u}(\mathbb{C}) .\nonumber \eeq
\end{notation}

In \cite{CS-Lim02}, we showed that for a link $L \equiv \{l^u:\ u=1, \cdots, \on\} \subset \bR^3$, $q=(q_1,\cdots, q_{\bar{m}})$,
\begin{align}
Z(q;& \{l^u, \varrho_u\}_{u=1}^{\on}) \nonumber \\
&=\sum_{g \in S(L)} \prod_{p \in \pd(\Sigma_i; L)}R(p)^{g(e^+(p)), g(\bar{e}^+(p)) }_{g(e^-(p)), g(\bar{e}^-(p)) }\prod_{p \in {\rm TDP}(\Sigma_i; L)}T(p)^{g(e^+(p)) }_{g(e^-(p)) }, \label{e.w.3}
\end{align}
whereby $S(L)$ denote all mappings $g$, such that
\begin{itemize}
  \item $g: E(\Sigma_i; L) \rightarrow \{1,2,\ldots, n\}$, $n$ is the maximum of all the $n_u$'s;
  \item for each $u$,\beq g|_{E(\Sigma_i; l^u)}: E(\Sigma_i; l^u) \rightarrow \{1,2,\ldots, n_u\}. \nonumber \eeq
\end{itemize}
Note that the Wilson Loop observable $Z(q; \{l^u, \varrho_u\}_{u=1}^{\on})$ is independent of any plane $\Sigma_i$, $i=1,2,3$, we choose to project onto to form a link diagram. This will in turn give us a state model for links, as discussed in \cite{CS-Lim02}.

\begin{rem}\label{r.r.1}
It is not true that the Wilson Loop observable given in Equation (\ref{e.w.3}) is invariant under ambient isotopy in $\bR^3$. To be a link invariant, it has to be invariant under 3 Reidemeister Moves. As such, this will impose a quantization condition on the charge $q_a$. Depending on the gauge group used, the quantization of the charge differs. When the gauge group $G_1$ is abelian, there is no restriction on $q_1$.

In the case of the fundamental representation on ${\rm SU}(N)= {\rm G}_a$, the charge squared $q_a^2$ must be an integer or half-integer, and the Wilson Loop observable will give us Homfly-type polynomial invariants. When $q_a^2$ is an even integer, we see that the Wilson Loop observable essentially gives us the linking number.

In the case of the fundamental representation on ${\rm SO}(4N)= {\rm G}_a$, the charge $q_a^2$ must be an odd integer, and the Wilson Loop observable will give us Conway-type polynomial invariants. In both cases, we see that the Chern-Simons path integral will quantize the charge. Refer to \cite{CS-Lim03} for the details.
\end{rem}

There is no sensible notation of a linking number for hyperlinks in $\bR^4$. In \cite{EH-Lim06}, we defined a time-like isotopic relation between time-like hyperlinks.  Any link invariant defined on $\pi_0(L)$, will be a time-like isotopic equivalence invariant for $L$.

\begin{notation}\label{n.l.1}
Recall in Section \ref{s.sum}, we have a time-like hyperlink $\chi(\oL, \uL)$, whereby $\oL = \{\ol^u:\ u=1, \cdots, \on\}$ is a matter hyperlink, coloured with representation $(\varrho_u, \rho_u)$ for each component loop $\ol^u$.

Define a set $\overline{T} = \{  \otau(1), \otau(2), \otau(3)\}$ containing 3 symbols and $V = \{1,2, \cdots, \on\}$. Consider the direct product $\overline{T} \times V$ and let $\mu = (\otau(i), u) \in \overline{T} \times V$, whereby $i=1,2,3$. Further define for $i=1,2,3$, $\overline{T}(i) \times V := \{(\otau(i), u): u \in V\}$.

Project $\pi_0[\chi(\ol^u, \uL)]$ onto $\Sigma_i$ to form a link diagram for each $i=1,2,3$. In future, it is useful to think of the matter hyperlink $\oL$, as a set of links $\oL \equiv \{ \pi_\mu(\oL):\ \mu \in \overline{T} \times V\}$, whereby when $\mu = (\otau(i), u)$, then $\pi_\mu(\oL)$ means we project $\ol^u$ onto a knot $\pi_0(\ol^u) \subset \bR^3$, and it defines a directed subgraph in $\Sigma_i$ as described in Item \ref{d.c.1.7} of Definition \ref{d.c.1}, and we have $E(\pi_\mu(\oL)):= E(\Sigma_i; \pi_0(\ol^u))$ denoting the set of edges in the directed subgraph on the plane $\Sigma_i$. The matter loop $\ol^\mu$ will be colored with representation $(\varrho_u, \rho_u)$.

Refer to Notation \ref{n.q.4}. If $\mu = (\otau(i), u)$, then we will let \beq \pd(\pi_\mu(\oL), \pi_0(\uL)) :=
\pd(\Sigma_i; \pi_0(\ol^u), \pi_0(\uL)).
\nonumber \eeq
\end{notation}

As a consequence, we see that for $i=1,2,3$,
\begin{align*}
Z(q; \{\pi_0(\ol^u), \varrho_u\}_{u=1}^{\on}) &=: Z(q; \{\ol^\mu, \varrho_\mu\}_{\mu \in \overline{T}(i) \times V}),
\end{align*}
for any $i=1,2,3$. It is understood that $\varrho_\mu = \varrho_u$ if $\mu = (\otau(i), u)$.

\section{Einstein-Hilbert action}

In this section, we summarise the definitions, notations and main results, taken from \cite{EH-Lim02, EH-Lim03}, which the reader should refer to for details.

Let $\mathfrak{su}(2)$ be the Lie Algebra of ${\rm SU}(2)$. Consider the following Lie Algebra $\mathfrak{su}(2) \times \mathfrak{su}(2)$ as described in \cite{EH-Lim03}. This direct product of Lie Algebras inherits the Lie bracket from $\mathfrak{su}(2)$ in the obvious way.

Let $\{\breve{e}_1, \breve{e}_2, \breve{e}_3\}$ be any basis for the first copy of $\mathfrak{su}(2)$ and
$\{\hat{e}_1, \hat{e}_2, \hat{e}_3\}$ be any basis for the second copy of $\mathfrak{su}(2)$, satisfying the conditions
\begin{align*}
[\breve{e}_1, \breve{e}_2] =& \breve{e}_3,\ \ [\breve{e}_2, \breve{e}_3] = \breve{e}_1,\ \ [\breve{e}_3, \breve{e}_1] = \breve{e}_2, \\
[\hat{e}_1, \hat{e}_2] =& \hat{e}_3,\ \ [\hat{e}_2, \hat{e}_3] = \hat{e}_1,\ \ [\hat{e}_3, \hat{e}_1] = \hat{e}_2.
\end{align*}

Using this basis, define \beq \mathcal{E}^+ = \sum_{i=1}^3\breve{e}_{i}\ ,\ \mathcal{E}^- = \sum_{i=1}^3\hat{e}_{i}. \nonumber \eeq

Let
\beq \hat{E}^{01} = (\breve{e}_1,0),\ \hat{E}^{02} = (\breve{e}_2,0) ,\ \hat{E}^{03} = (\breve{e}_3,0)  \nonumber \eeq and \beq \hat{E}^{23} = (0,\hat{e}_1) ,\ \hat{E}^{31} = (0,\hat{e}_2) ,\ \hat{E}^{12} = (0,\hat{e}_3). \nonumber \eeq Do note that $\hat{E}^{\alpha\beta} = -\hat{E}^{\beta\alpha}$.

Finally, denote \beq \mathcal{E} :=
\left(
  \begin{array}{cc}
    \mathcal{E}^+ &\ 0 \\
    0 &\ -\mathcal{E}^- \\
  \end{array}
\right) \cong (\mathcal{E}^+, -\mathcal{E}^-) \in \mathfrak{su}(2) \times \mathfrak{su}(2). \nonumber \eeq

For $\kappa > 0$, let $\overline{\mathcal{S}}_\kappa(\bR^4) \subset L^2(\bR^4)$ be a Schwartz space, which consists of functions of the form $f = p \cdot \sqrt{\phi_\kappa}$, whereby $p$ is a polynomial on $\bR^4$, and $\phi_\kappa(\vec{x}) = \frac{\kappa^4}{(2\pi)^2}e^{-\kappa^2|\vec{x}|^2/2}$ is the Gaussian function with variance $1/\kappa^2$. Suppose $V \rightarrow \bR \times \bR^3$ is a 4-dimensional trivial vector bundle.

Define
\begin{align}
\begin{split}
L_\omega^\kappa :=& \overline{\mathcal{S}}_\kappa(\bR^4) \otimes \Lambda^1(\bR^3)\otimes [\mathfrak{su}(2) \times \mathfrak{su}(2)], \\
L_e^\kappa :=& \overline{\mathcal{S}}_\kappa(\bR^4) \otimes \Lambda^1(\bR^3)\otimes V.
\end{split}\label{e.o.1}
\end{align}

The dynamical variables of General Relativity are given by $e$ and $\omega$, which are $V$-valued one form and spin connection respectively. A spin connection $\omega$ on $V$, can be written as $\omega \equiv A^a_{\alpha\beta} \otimes dx_a\otimes \hat{E}^{\alpha\beta}$, whereby $A^a_{\alpha\beta}: \bR^4 \rightarrow \bR$ is smooth and $E^{\alpha\beta} \in \mathfrak{su}(2) \times \mathfrak{su}(2)$. Under axial gauge fixing, every spin connection can be gauge transformed into \beq \omega = A^i_{\alpha\beta} \otimes dx_i \otimes \hat{E}^{\alpha\beta}\in \overline{\mathcal{S}}_\kappa(\bR^4) \otimes \Lambda^1(\bR^3)\otimes [\mathfrak{su}(2) \times \mathfrak{su}(2)] \equiv L_\omega^\kappa,\nonumber \eeq $A^i_{\alpha\beta}: \bR^4 \rightarrow \bR$ smooth, subject to the conditions \beq A^1_{\alpha\beta}(0, x_1, 0, 0) = 0,\ A^2_{\alpha\beta}(0, x_1, x_2, 0)=0,\ A^3_{\alpha\beta}(0, x_1, x_2, x_3) = 0. \label{e.r.1} \eeq There is an implied sum over repeated indices, $i$ runs from 1 to 3. We will drop all these restrictions on $A^i_{\alpha\beta}$.

\begin{rem}\label{r.lqg.1}
Note that $A^{i}_{\alpha\beta} = -A^{i}_{\beta\alpha}\in \overline{\mathcal{S}}_\kappa(\bR^4)$.
\end{rem}

Let $\{E^\gamma\}_{\gamma=0}^3$ be a basis for $V$. After axial gauge fixing, we only consider $e: T\bR^4 \rightarrow V$, of the form \beq e = B^i_\gamma \otimes dx_i \otimes E^\gamma \in \overline{\mathcal{S}}_\kappa(\bR^4) \otimes \Lambda^1(\bR^3) \otimes V  \equiv L_e^\kappa. \nonumber \eeq There is an implied sum over repeated indices, $i$ runs from 1 to 3.

\begin{rem}
Applying axial gauge fixing to both $\omega$ and $e$, is analogous to applying a Gauss constraint, as explained in \cite{EH-Lim07}.
\end{rem}

The curvature of $\omega$ is given by $R = d\omega + \omega \wedge \omega$, and the Einstein-Hilbert action is given by \beq S_{{\rm EH}}(e, \omega) = \frac{1}{8}\int_{\bR^4}e \wedge e \wedge R. \nonumber \eeq In terms of the variables $\{A_{\alpha \beta}^i, B_\gamma^j\}$, the Einstein-Hilbert action is written as
\begin{align*}
S_{{\rm EH}}(e, \omega) =& \frac{1}{8}\int_{\bR^4}\epsilon^{abcd}B^1_\gamma B^2_\mu[E^{\gamma \mu}]_{ab} \cdot \partial_0 A^3_{\alpha\beta}[E^{\alpha\beta}]_{cd} dx_1\wedge dx_2 \wedge dx_0 \wedge dx_3\\
+& \frac{1}{8}\int_{\bR^4}\epsilon^{abcd}B^2_\gamma B^3_\mu[E^{\gamma \mu}]_{ab} \cdot \partial_0 A^1_{\alpha\beta}[E^{\alpha\beta}]_{cd} dx_2\wedge dx_3 \wedge dx_0 \wedge dx_1\\
+&\frac{1}{8}\int_{\bR^4}\epsilon^{abcd}B^3_\gamma B^1_\mu[E^{\gamma \mu}]_{ab} \cdot \partial_0 A^2_{\alpha\beta}[E^{\alpha\beta}]_{cd} dx_3\wedge dx_1 \wedge dx_0 \wedge dx_2,
\end{align*}
which is an invariantly defined integral. There is an implied sum over repeated indices. And $\epsilon^{\mu \gamma \alpha \beta} \equiv \epsilon_{\mu \gamma \alpha \beta}$ is equal to 1 if the number of transpositions required to permute $(0123)$ to $(\mu\gamma\alpha\beta)$ is even; otherwise it takes the value -1.

\begin{rem}
Note that $[E^{\alpha\beta}]_{ab} = 1$ if $\alpha=a$, $\beta = b$; $[E^{\alpha\beta}]_{ab} = -[E^{\beta\alpha}]_{ab} = -1$ if $\alpha=b$, $\beta = a$; $[E^{\alpha\beta}]_{ab}= 0$ for all other cases.
\end{rem}

Recall from Section \ref{s.sum}, we color each component of the matter hyperlink with a representation. This means choose a representation $\rho_u\equiv (\rho_u^+, \rho_u^-): \mathfrak{su}(2) \times \mathfrak{su}(2) \rightarrow {\rm End}(V_u^+) \times {\rm End}(V_u^-)$ for each component loop $\ol^u$, $u=1, \ldots, \on$, in the hyperlink $\oL$. Therefore,
\beq \rho_u(\mathcal{E}) =
\left(
  \begin{array}{cc}
    \rho_u^+(\mathcal{E}^+) &\ 0 \\
    0 &\ -\rho_u^-(\mathcal{E}^+) \\
  \end{array}
\right). \nonumber \eeq
Note that we do not color the components loops in $\uL$, i.e. we do not choose a representation for $\uL$. When either $\rho_u^+$ or $\rho_u^-$ is zero, we will write $\rho_u = \rho_u^-$ and $\rho_u = \rho_u^+$ respectively.

\begin{rem}\label{r.l.3}
We will choose $\rho_u^\pm: \mathfrak{su}(2) \rightarrow {\rm End}(V_u^\pm)$ to be an irreducible representation. Note that the dimension of $V_u^\pm$ is given by $2j_{\rho_u^\pm} + 1$, $j_{\rho_u^\pm} \geq 0$ is an half-integer or an integer.  Then it is known that the Casimir operator is
\begin{align*}
\sum_{i=1}^3 \rho^+(\hat{E}^{0i})\rho^+(\hat{E}^{0i})  = -\xi_{\rho^+} I_{\rho^+},\\
\sum_{i=1}^3 \rho^-(\hat{E}^{\tau(i)})\rho^-(\hat{E}^{\tau(i)}) = -\xi_{\rho^-} I_{\rho^-},
\end{align*}
$I_{\rho^\pm}$ is the $2j_{\rho^\pm} + 1$ identity operator on $V^\pm$ and $\xi_{\rho^\pm} := j_{\rho^\pm}(j_{\rho^\pm}+1)$.
\end{rem}

Let $q_0$ be some real constant, known as the charge. Define
\begin{align}
\begin{split}
V(\{\ul^v\}_{v=1}^{\underline{n}})(e) :=& \exp\left[ \sum_{v=1}^{\un} \int_{\ul^v} \sum_{\gamma=0}^3 B^i_\gamma \otimes dx_i\right], \\
W(q_0; \{\ol^u, \rho_u\}_{u=1}^{\on})(\omega) :=& \prod_{u=1}^{\on}\Tr_{\rho_u}  \mathcal{T} \exp\left[ q_0\int_{\ol^u} A^i_{\alpha\beta} \otimes dx_i\otimes \hat{E}^{\alpha\beta}  \right].
\end{split}\label{e.l.7}
\end{align}
Here, $\mathcal{T}$ is the time-ordering operator as defined in Definition \ref{d.t.1}. And we sum over repeated indices, with $i$ taking values in 1, 2 and 3.

A loop representation is an Einstein-Hilbert path integral, of the form \beq \frac{1}{Z_{{\rm EH}}^\kappa}\int_{\omega \in L_\omega^\kappa,\ e \in L_e^\kappa}V(\{\ul^v\}_{v=1}^{\un})(e) W(q_0; \{\ol^u, \rho_u\}_{u=1}^{\on})(\omega)  e^{i S_{{\rm EH}}(e, \omega)}\ D[e] D[\omega], \label{ex.eh.1} \eeq whereby $D[e]$ and $D[\omega]$ are non-existent Lebesgue measures on $L_e^\kappa$ and $L_\omega^\kappa$ respectively and \beq Z_{{\rm EH}}^\kappa = \int_{\omega \in L_\omega^\kappa,\ e \in L_e^\kappa}e^{i S_{{\rm EH}}(e, \omega)}\ D[e] D[\omega]. \label{e.l.8} \eeq This path integral is indexed by $\kappa$, with $i = \sqrt{-1}$, and made sense of in \cite{EH-Lim02} using a set of Chern-Simon rules. By taking $\kappa$ going to infinity, it allows us to define a functional, termed as a Wilson Loop observable $Z(q_0; \{l^u, \rho_u\}_{u=1}^{\on}, \uL )$, of a colored hyperlink $\chi(\oL, \uL)$.

Explicitly, this Wilson Loop observable is given by
\begin{align}
Z&(q_0; \{l^u, \rho_u\}_{u=1}^{\on}, \uL)\nonumber \\
&:=\prod_{u=1}^{\on} \left( \Tr_{\rho^+_u}\ \exp[-\pi iq_0\ {\rm sk}(\ol^u, \uL) \cdot \mathcal{E}^+] +
\Tr_{\rho^-_u}\ \exp[\pi iq_0\ {\rm sk}(\ol^u, \uL) \cdot \mathcal{E}^-] \right). \label{e.weh.1}
\end{align}
Here, $\Tr_{\rho^\pm_u}$ is the matrix trace in the representation $\rho_u^\pm$ for $\mathfrak{su}(2)$. It depends on the hyperlinking number between each pair of matter and geometric loop $\ol^u$ and $\ul^v$, and the representation $\rho_u$ for each component matter loop. The detailed computations can be found in \cite{EH-Lim03}.

\begin{rem}
To compute the Wilson Loop observable, we only use the projection $\pi_0(\chi(\oL, \uL))$, and project it onto $\Sigma_i$, i=1,2,3, using link diagrams. Refer to \cite{EH-Lim03} for details.
\end{rem}

\begin{defn}\label{d.z.1}
We will write
\begin{align*}
Z^\pm(q_0; \{l^u, \rho_u\}_{u=1}^{\on}, \uL ) &:= \prod_{u=1}^{\on}\Tr_{\rho^\pm_u}\ \exp\left[\mp\pi iq_0\ {\rm sk}(\ol^u, \uL) \cdot \mathcal{E}^\pm \right].
\end{align*}
\end{defn}

\begin{notation}\label{n.q.1}
Refer to Notation \ref{n.q.4} and Definition \ref{d.c.2}. For \\
$p \in \bigcup_{u=1}^{\on}\bigcup_{k=1}^3 \pd(\Sigma_k; \pi_0(\ol^u), \pi_0(\uL))$, write \beq \tilde{Q}(p) := \exp\left[ -\pi i q_0\ \sigma(p)\rho_u(\mathcal{E})\right], \nonumber \eeq the crossing $p$ is also identified by 2 edges $(e^+(p), e^-(p))$, which lie in $E(\Sigma_k;\pi_0(\ol^u))$ for a unique $k$ and $u$.
\end{notation}

By Equation (\ref{e.sk.1}), we can write $Z(q_0; \{l^u, \rho_u\}_{u=1}^{\on}, \uL )$ as
\begin{align}
\prod_{u=1}^{\on} &\left( \Tr_{\rho_u^+}\ \exp[-\pi iq_0\ {\rm sk}(\ol^u, \uL) \cdot \mathcal{E}^+] +
\Tr_{\rho_u^-}\ \exp[\pi iq_0\ {\rm sk}(\ol^u, \uL) \cdot \mathcal{E}^-] \right) \nonumber \\
&= \prod_{u=1}^{\on}\left[ \Tr\ \prod_{p \in \bigcup_{k=1}^3 \pd(\Sigma_k; \pi_0(\ol^u), \pi_0(\uL))}\tilde{Q}(p) \right] \label{e.n.1}.
\end{align}
Because $\{\exp\left[ \pi i q_0\ \sigma(p)\rho_u(\mathcal{E})\right]:\ \sigma(p) = \pm 1\}$ commute, we see that the time-ordering operator $\mathcal{T}$ is no longer necessary.

\begin{rem}\label{r.l.2}
Note that Equation (\ref{e.weh.1}) is not invariant under time-like isotopy. We need to time-order the matter loop and geometric loop. See \cite{EH-Lim06}. By doing so,
\begin{align}
\Tr_{\rho_u^\pm}\ \exp\left[-\pi iq_0\ {\rm sk}(\ol^u, \ul^v) \cdot \mathcal{E}^\pm \right] &= \Tr_{\rho_u^\pm}\ \exp\left[3\alpha\pi iq_0\ {\rm lk}(\pi_0(\ol^u), \pi_0(\ul^v)) \cdot \mathcal{E}^\pm\right], \label{e.a.4}
\end{align}
whereby ${\rm lk}(\pi_0(\ol^u), \pi_0(\ul^v))$ denotes the linking number between oriented knots $\pi_0(\ol^u)$ and $\pi_0(\ul^v)$, and $\alpha$ takes values $\pm 1$, depending on the time-ordering. Also see Remark \ref{r.l.4}. As a result, Equation (\ref{e.weh.1}) is dependent on the linking number. It is an invariant under time-like isotopy that preserves time-ordering. Observe that in the above expression, it only involves the linking numbers between each projected matter loop and a projected geometric loop, inside $\bR^3$.

In the rest of this article, we will always assume that each pair of matter component loop and geometric component loop is time-ordered.
\end{rem}

\section{Unification}\label{s.uni}

\begin{notation}
Write $\vec{q} = (q, q_0) := (q_1, \cdots, q_{\bar{m}}, q_0) \in \bR^{\bar{m} + 1}$ be fixed throughout this article.
\end{notation}

\begin{notation}
Refer to Notation \ref{n.l.1}. Let $p_1, p_2, p_3$ be indices such that $p_1$, $p_2$ and $p_3$ take values in $\{2,3\} \equiv \otau(1)$, $\{1,3\} \equiv \otau(2)$ and $\{1,2\} \equiv \otau(3)$ respectively.

Define for $i=1,2,3$,
\begin{align*}
K_{{\otau}(i)}^{\kappa} :=& \left\{A_{{\otau}(i),\alpha}^{a,p_i} \otimes dx_{p_i} \otimes E_a^\alpha:\ A_{{\otau}(i),\alpha}^{a,p_i} \in \overline{\mathcal{S}}_\kappa(\bR^3) \right\}.
\end{align*}
There is an implicit sum over $\alpha$, $p_i$ and $a$, each $\alpha$ takes values from 1 to $N_a$, $p_i$ taking values in $\otau(i)$, and $a=  1, \cdots, \bar{m}$.

For each $i=1,2,3$, $A_{{\otau}(i),\alpha}^{a, p_i} \in \overline{\mathcal{S}}_\kappa(\bR^3)$ for $\kappa > 0$ fixed. We will write $A^{{\otau}(i)} = A_{{\otau}(i), \alpha}^{a, p_i} \otimes dx_{p_i} \otimes E_a^\alpha \in K_{{\otau}(i)}^{\kappa}$. Note its dependence on the parameter $\kappa$.

Further define
\begin{align*}
S_{{\rm CS}}(A^{\otau(i)}) = \frac{\varsigma}{4\pi} \int_{\mathbb{R}^3} \Tr \left[A^{\otau(i)} \wedge dA^{\otau(i)} + \frac{2}{3}A^{\otau(i)} \wedge A^{\otau(i)} \wedge A^{\otau(i)}\right],\ \varsigma \neq 0.
\end{align*}

Let $A^{\otau} = \{A^{{\otau}(i)}:\ i=1,2,3\}$. Note that $A^{\otau}$ defines a $\mathfrak{g}$-valued connection $\sum_{i=1}^3 A^{{\otau}(i)}$ on $\bR \times \bR^3$.
\end{notation}

In Equation (\ref{e.a.4}), we saw that only the linking number between a pair of projected matter and projected geometric loops appears in the Wilson Loop observable. The linking number between pairs of projected component matter loops do not appear in the formula. Furthermore, it does not give us Homfly-type of knot invariants. It would seem that the Einstein-Hilbert theory is not complete.

Now suppose we try to merge Chern-Simons theory and Einstein-Hilbert theory together. That is, we consider the gauge group ${\rm G} \times [{\rm SU}(2) \times {\rm SU}(2)]$. The action we want to consider is \beq S(A^{\otau}, e, \omega) := \sum_{i=1}^3  S_{{\rm CS}}(A^{\otau(i)}) + S_{{\rm EH}}(e, \omega). \nonumber \eeq  Write
\begin{align*}
\Omega^\kappa &:= \left(\{K_{\otau(i)}^{\kappa}\}_{i=1}^3, L_\omega^\kappa, L_e^\kappa \right),\quad
D[A^{\otau}, e, \omega] = \left[\prod_{i=1}^3D[A^{\otau(i)}] \right] D[e] D[\omega],
\end{align*}
whereby $D[A^{\otau}, e, \omega]$ is some non-existent Lebesgue measure.

\begin{notation}\label{n.n.2}
Recall in Section \ref{s.sum}, we defined a matter (geometric) time-like hyperlink $\oL$ ($\uL$), and both hyperlinks tangled together to form $\chi(\oL, \uL)$.

Refer to Equations (\ref{e.b.1}) and (\ref{e.l.7}). Write for $i=1,2,3$,
\begin{align*}
S(\{q_a\}_{a=1}^{\bar{m}}, \ol^u, \varrho_u)(A^{{\otau}(i)})
&:=
\exp\left[\sum_{a=1}^{\bar{m}} q_a\int_{\pi_0(\ol^u)} A_{{\otau}(i),\alpha}^{a, p_i} \otimes dx_{p_i} \otimes \varrho_u(E_a^\alpha) \right], \\
\bar{S}(q_0, \ol^u, \rho_u)(\omega) &:= \exp\left[ \frac{q_0}{3}\int_{\ol^u} A^j_{\alpha\beta} \otimes dx_j\otimes \rho_u(\hat{E}^{\alpha\beta})\right].
\end{align*}
Thus, define
\begin{align*}
W(\vec{q}; \ol^u, \varrho_u, \rho_u)(A^{\otau(i)}, \omega) &=
\Tr\  \mathcal{T}
\left(
  \begin{array}{cc}
    S(\{q_a\}_{a=1}^{\bar{m}}, \ol^u, \varrho_u)(A^{{\otau}(i)}) &\  0 \\
     0&\ \bar{S}(q_0, \ol^u, \rho_u)(\omega)  \\
  \end{array}
\right).
\end{align*}

And \beq W(\vec{q}; \{\ol^u, \varrho_u, \rho_u \}_{u=1}^{\on})(A^{\otau}, \omega) := \prod_{u=1}^{\on} W(\vec{q}; \ol^u, \varrho_u, \rho_u)(A^{\otau}, \omega), \nonumber
\eeq for
\begin{align*}
W(\vec{q};\ &\ol^u, \varrho_u, \rho_u)(A^{\otau}, \omega) := \prod_{i=1}^3 W(\vec{q}; \ol^u, \varrho_u, \rho_u)(A^{\otau(i)}, \omega).
\end{align*}

We will write \beq Z_{{\rm CS}}^{\kappa, \otau(i)} = \int_{A^{\otau(i)} \in K_{{\otau}(i)}^{\kappa}} e^{iS_{{\rm CS}}(A^{\otau(i)})}DA^{\otau(i)},\ {\rm and}\ Z_{{\rm CS}}^\kappa  = \prod_{i=1}^3 Z_{{\rm CS}}^{\kappa, \otau(i)}.\nonumber \eeq
\end{notation}

We want to make sense of the following path integral
\beq \frac{1}{Z^\kappa}\int_{(A^{\otau}, \omega, e)\in \Omega^\kappa}V(\{\ul^v\}_{v=1}^{\un})(e) W(\vec{q}; \{\ol^u, \varrho_u, \rho_u\}_{u=1}^{\on})(A^{\otau}, \omega)  e^{i S(A^{\otau}, e, \omega)}\ D[A^{\otau}, e, \omega], \label{ex.toe.1} \eeq whereby $Z^\kappa = Z_{{\rm CS}}^\kappa Z_{{\rm EH}}^\kappa$, and compute its limit as $\kappa \rightarrow \infty$. Note that $V(\{\ul^v\}_{v=1}^{\un})(e)$ and $Z_{{\rm EH}}^\kappa$ were defined in Equations (\ref{e.l.7}) and (\ref{e.l.8}) respectively.

\begin{notation}\label{n.n.1}
We will write for $i=1,2,3$,
\begin{align*}
W_u(A^{\otau(i)})
&:=
\left(
  \begin{array}{cc}
    S(\{q_a\}_{a=1}^{\bar{m}}, \ol^u,\varrho_u)(A^{{\otau}(i)}) &\ 0 \\
    0 &\ I_{j_{\rho_u} } \\
  \end{array}
\right), \\
W_u(\omega)
&:=
\left(
  \begin{array}{cc}
    I_{n_u} &\ 0 \\
    0 &\ \bar{S}(q_0, \ol^u, \rho_u)(\omega) \\
  \end{array}
\right),
\end{align*}
which commute with each other. Here, $I_{n_u}$ and $I_{j_{\rho_u}}$ denotes the $n_u$ and $2[j_{\rho_u}^+ + j_{\rho_u}^- +1]$ identity matrices respectively.
\end{notation}

\begin{rem}\label{r.r.2}
Note that \beq W_u(A^{\otau(i)})W_u(\omega) =
\left(
  \begin{array}{cc}
    S(\{q_a\}_{a=1}^{\bar{m}}, \ol^u,\varrho_u)(A^{{\otau}(i)}) &\ 0 \\
    0 &\ \bar{S}(q_0, \ol^u, \rho_u)(\omega) \\
  \end{array}
\right)
. \nonumber \eeq
We can regard $\mathcal{T}[ W_u(A^{\otau(i)})W_u(\omega)]$, as the holonomy operator along the loop $\ol^u$, given the $\mathfrak{g} \times [\mathfrak{su}(2) \times \mathfrak{su}(2)]$-valued connection $\{A^{\otau(i)}, A^{j}_{\alpha\beta} \otimes dx_{j}\otimes \rho(\hat{E}^{\alpha\beta})\}$, for $i=1,2,3$.
\end{rem}

\begin{defn}($\widetilde{\Tr}$)\label{d.tr.1}\\
Recall $\overline{T} \times V$ from Notation \ref{n.l.1}. Label each elements in $\overline{T} \times V$ uniquely using elements in the set $\{1, 2, \cdots, 3\on\}$, in any order.

Define a linear functional $\widetilde{\Tr}$ as follows. Suppose a matrix $A$ is indexed by time $s$, and by component $\theta$, $\theta \in \{1, 2, \cdots, 3\on\}$. In other words, $A \equiv A(\theta ,s)$. Let $\{A(\pi_1,s_1),\ldots, A(\pi_n,s_n)\}$ be a finite set of matrices. Let
\begin{align*}
S_{\theta} &= \{v \in \{1,\ldots,n \}: \pi_v = \theta\},
\end{align*}
and write $m_{\theta} := |S_{\theta}|$. For any $n \geq 1$, define a linear operator,
\begin{align*}
&\widetilde{\Tr}: A(\pi_1,s_1) \otimes \cdots \otimes A(\pi_n,s_n)\longmapsto  \\
&  \Tr[A( 1,s_{\beta_1(1)})\cdots \otimes A(1,s_{\beta_1(m_1)}) ]  \cdots \Tr [A(r,s_{\beta_r(1)}) \cdots A(r,s_{\beta_r(m_r)})],
\end{align*}
such that for each $i=1,\ldots, r$, $s_{\beta_i(1)} > s_{\beta_i(2)} > \ldots > s_{\beta_i(m_i)}$ and $\beta_\theta(j) \in S_\theta$ for $j=1,\ldots, m_\theta$.
\end{defn}

\begin{notation}\label{n.q.3}
Refer to Definition \ref{d.c.3} and Notations \ref{n.q.2} and \ref{n.q.1}. Here, we define the $R$ matrix as \beq R(p) =
\left(
  \begin{array}{cc}
    \exp\left[-\varepsilon(p)\sum_{a=1}^{\bar{m}} \pi i q_a^2 \sum_{\alpha=1}^{N_a} \varrho_u(E_a^\alpha)\otimes \varrho_v(E_a^\alpha) \right] &\ 0 \\
    0 &\ I_{j_{\rho_u}} \otimes I_{j_{\rho_v}} \\
  \end{array}
\right). \nonumber \eeq And
\beq T(p) =
\left(
  \begin{array}{cc}
    \exp\left[-\varepsilon(p)\sum_{a=1}^{\bar{m}}\frac{\pi i q_a^2}{2}\sum_{\alpha=1}^{N_a} \varrho_u(E_a^\alpha)\cdot \varrho_u(E_a^\alpha) \right] &\ 0 \\
    0 &\ I_{j_{\rho_u}} \\
  \end{array}
\right), \nonumber \eeq with
\beq Q(p) =
\left(
  \begin{array}{cc}
    I_{n_u} &\ 0 \\
    0 &\ \exp\left[ \pi i q_0\ \sigma(p)\rho_u(\mathcal{E})\right] \\
  \end{array}
\right). \nonumber \eeq

Recall that $\varrho_u: \mathfrak{g} \rightarrow {\rm End}(\bC^{n_u})$ is an irreducible representation, and $\rho_u \equiv (\rho_u^+, \rho_u^-)$ is given by a pair of irreducible representations of $\mathfrak{su}(2) \times \mathfrak{su}(2)$. For each $u=1, \cdots, \on$, let $m_u^\pm$ be the dimension of representation $\rho_u^\pm$. And write $m_u := m_u^+ + m_u^-$.

Let $n$ and $m$ be the maximum of all the $n_u$'s and $m_u$'s respectively.
\end{notation}

\begin{thm}\label{t.u.1}
Recall we defined $\pd(\Sigma_k, \pi_0(\oL))$ in Notation \ref{n.q.4} and the product set $\overline{T} \times V$ in Notation \ref{n.l.1}. The component matter loops in $\oL$ will be labelled as $\{\ol^\mu,\varrho_\mu, \rho_\mu: \mu \in \overline{T} \times V\}$ and it is understood that $\varrho_\mu = \varrho_u$, $\rho_\mu = \rho_u$ if $\mu = (\otau(i), u)$.

Suppose $\mathcal{S}$ is $\overline{T} \times V$.
Define for $\mathcal{S}$,
\begin{align}
Z(\vec{q}&;\{\ol^\mu, \varrho_\mu, \rho_\mu\}_{\mu \in \mathcal{S}}, \uL ) \nonumber \\
&:= \sum_{g \in S(\mathcal{\oL})} \prod_{p \in U_1}R(p)^{g(e^+(p)), g(\bar{e}^+(p)) }_{g(e^-(p)), g(\bar{e}^-(p)) }\prod_{p \in U_2}T(p)^{g(e^+(p)) }_{g(e^-(p)) }\prod_{p \in U_3}Q(p)^{g(e^+(p)) }_{g(e^-(p)) }, \label{e.z.2}
\end{align}
whereby (Refer to Notations \ref{n.q.4} and \ref{n.l.1}, and Item \ref{i.f.1} of Definition \ref{d.c.1}.)
\begin{align}
\label{e.u.1}
\begin{split}
U_1 &= \bigcup_{i=1}^3\pd(\Sigma_i; \pi_0(\oL)), \\
U_2 &= \bigcup_{i=1}^3{\rm TDP}(\Sigma_i; \pi_0(\oL)) , \quad
U_3 =  \bigcup_{\mu \in \mathcal{S}}\pd(\pi_\mu(\oL), \pi_0(\uL)),
\end{split}
\end{align}
and $S(\oL)$ denote all mappings $g$, such that
\begin{itemize}
  \item $g: \bigcup_{\mu \in \mathcal{S}}E(\pi_\mu(\oL)) \longrightarrow \{1,2,\ldots, n+m\}$;
  \item for each $\mu= (\otau(i), u) \in \mathcal{S}$, \beq g|_{E(\pi_\mu(\oL))}: E(\pi_\mu(\oL)) \longrightarrow \{1,2,\ldots, n_u+ m_u\}. \nonumber \eeq
\end{itemize}

Using the definitions of $Z(q, \{\ol^\mu, \varrho_\mu\}_{\mu \in \overline{T}(i) \times V})$ and $Z(q_0; \{\ol^u, \rho_u\}_{u=1}^{\on}, \uL )$, we see that the limit of Expression \ref{ex.toe.1} as $\kappa \rightarrow \infty$, is given by $Z(\vec{q};\{\ol^\mu, \varrho_\mu, \rho_\mu\}_{\mu \in \overline{T} \times V}, \uL )$. We will henceforth refer it as the Wilson Loop observable for the tangled time-like hyperlink $\chi(\oL, \uL)$.

\end{thm}

\begin{rem}\label{r.l.1}
\begin{enumerate}
  \item The above Wilson Loop observable takes into account, how each component matter loop is 'linked' with other component matter loops inside $\oL$, and the hyperlinking number between each pair of component matter loop and component geometric loop. It does not consider how each component geometric loop is `linked' inside $\uL$ with other component geometric loops.
  \item When we compute the Wilson Loop observable and other physical quantities in LQG, we never consider $\pi_i(\chi(\oL, \uL))$, but only consider projection using $\pi_0$. See \cite{EH-Lim03}.
  \item\label{r.l.1c} When $\uL = \emptyset$, then the Wilson Loop observable becomes
  \begin{align*}
\prod_{i=1}^3 Z(\vec{q};\{\ol^\mu, \varrho_\mu, \rho_\mu\}_{\mu \in \overline{T}(i) \times V}, \emptyset )
&:= \prod_{i=1}^3 Z(q;\{\ol^\mu, \varrho_\mu\}_{\mu \in \overline{T}(i) \times V}).
\end{align*}
We will refer it as a quantum Chern-Simons invariant for a time-like hyperlink.
\end{enumerate}
\end{rem}

\begin{proof}
Refer to Notation \ref{n.n.1}. We will write $B^{\otimes^k} = B \otimes B \otimes \cdots \otimes B$, $k$ copies of $B$. Write \beq  W(A^{{\otau}(i)}) = \bigotimes_{u=1}^{\on}W_u(A^{{\otau}(i)}), \quad W(\omega) =
\bigotimes_{u=1}^{\on}W_u(\omega), \label{e.w.6} \eeq and note that
\beq \bigotimes_{i=1}^3 W_u(A^{{\otau}(i)})W_u(\omega) = \left[ \bigotimes_{i=1}^3 W_u(A^{{\otau}(i)})\right] W_u(\omega)^{\otimes^3}. \nonumber \eeq
Hence from Remark \ref{r.r.2}, we can write
Expression \ref{ex.toe.1} as \beq \widetilde{\Tr}\bigotimes_{i=1}^3 Y_i \otimes \frac{1}{Z_{{\rm EH}}}\int_{\omega, e} W(\omega)^{\otimes^3}V(\{\ul^v\}_{v=1}^{\un})(e) e^{iS_{{\rm EH}}(e,\omega)} D[e]D[\omega], \label{ex.toe.2} \eeq where
\begin{align*}
Y_i^\kappa &= \frac{1}{Z_{{\rm CS}}^{\kappa, \otau(i)}}\int_{A^{{\otau}(i)}\in K_{\otau(i)}^\kappa} W(A^{{\otau}(i)}) e^{iS_{{\rm CS}}(A^{{\otau}(i)})} D[A^{{\otau}(i)}],
\end{align*}
indexed by $\kappa$, which we have omitted from the notation in Expression \ref{ex.toe.2}.

In \cite{CS-Lim02}, we explained how after taking into account of the self-linking problem by considering a frame, a Chern-Simons path integral Expression \ref{e.w.2} will give us
\begin{align*}
\lim_{\kappa \rightarrow \infty} \frac{1}{Z^{\cdot, \otau(i)}_{{\rm CS}}}\int_{A^{{\otau}(i)}} \bigotimes_{u=1}^{\on}W_u(A^{{\otau}(i)}) e^{iS_{{\rm CS}}(A^{{\otau}(i)})}& D[A^{{\otau}(i)}] \\
&=\bigotimes_{p \in \pd(\Sigma_i; \pi_0(\oL))}\mathcal{R}(p) \  \bigotimes_{p \in {\rm TDP}(\Sigma_i; \pi_0(\oL))}\mathcal{T}(p),
\end{align*}
whereby we write
\begin{align*}
\mathcal{R}(p)
:= \left(
  \begin{array}{cc}
   \tilde{E}(p)  &\ 0 \\
    0 &\ \tilde{G} \\
  \end{array}
\right), \quad
\mathcal{T}(p) :=
\left(
  \begin{array}{cc}
   \tilde{F}(p)  &\ 0 \\
    0 &\ \tilde{G} \\
  \end{array}
\right),
\end{align*}
with $\tilde{G} :=I_{j_{\rho_u}} \otimes I_{j_{\rho_v}}$, and
\begin{align*}
\tilde{E}(p) &:= \exp\left[-i\pi \sum_{a=1}^{\bar{m}}q_a^2 \varepsilon(p)\sum_{\alpha=1}^{N_a} \varrho_u(E^\alpha)(e^+(p), e^-(p)) \otimes \varrho_v(E^\alpha)(\bar{e}^+(p), \bar{e}^-(p)) \right], \\
\tilde{F}(p) &:= \exp\left[-\frac{i\pi}{2} \sum_{a=1}^{\bar{m}}q_a^2 \varepsilon(p)\sum_{\alpha=1}^{N_a} \varrho_u(E^\alpha)(e^+(p), e^-(p)) \otimes \varrho_u(E^\alpha)(e^+(p), e^-(p)) \right].
\end{align*}
Note that the edges $(e^+(p), e^-(p))$ ($(\bar{e}^+(p), \bar{e}^-(p))$) belong to the subgraph formed from projecting $\pi_0(\ol^u)$ ($\pi_0(\ol^v)$) on $\Sigma_i$. And $\varrho_u(E^\alpha)(e^+(p), e^-(p)) \equiv \varrho_u(E^\alpha)$, except that we use $\{e^+(p), e^-(p)\}$ to time-order the matrices, to be arranged according to $\widetilde{\Tr}$.

There is ambiguity in the ordering of the tensor products above, so it is not well-defined. However, when we apply the operator $\widetilde{\Tr}$ given in Definition \ref{d.tr.1}, the RHS of the above equation will be well-defined.

Refer to Equation (\ref{e.n.1}). By taking the limit as $\kappa$ goes to infinity, we showed in \cite{EH-Lim03}, the Einstein-Hilbert path integral Expression \ref{ex.eh.1} will give us
\begin{align*}
\lim_{\kappa \rightarrow \infty}& \frac{1}{Z^\kappa_{{\rm EH}}}\int_{\{\omega \in L_\omega^\kappa, e \in L_e^\kappa\}} \bigotimes_{u=1}^{\on}W_u(\omega)^{\otimes^3}V(\{\ul^v\}_{v=1}^{\un})(e)e^{iS_{{\rm EH}}(e,\omega)} D[e]D[\omega] = \bigotimes_{k=1}^3 \overline{B}(k),
\end{align*}
whereby \beq \overline{B}(k) = \bigotimes_{p \in \bigcup_{u=1}^{\on}\pd(\Sigma_k; \pi_0(\ol^u), \pi_0(\uL))}Q(p), \nonumber \eeq
because for $u=1, \cdots, \on$, each $p \in \bigcup_{k=1}^3\pd(\Sigma_k; \pi_0(\ol^u), \pi_0(\uL))$ is considered 3 times, due to taking the tensor product $W_u(\omega)^{\otimes^3}$.

Refer to Notation \ref{n.l.1}. Therefore, Expression \ref{ex.toe.2} becomes $\widetilde{\Tr}[ \overline{I} \otimes \overline{J}\otimes  \overline{K}]$, whereby
\begin{align}
\label{e.d.1}
\begin{split}
\overline{I} :=& \bigotimes_{i=1}^3\bigotimes_{p \in \pd(\Sigma_{i};\pi_0(\oL))}\mathcal{R}(p), \quad
\overline{J} :=
\bigotimes_{i=1}^3 \bigotimes_{p \in {\rm TDP}(\Sigma_i; \pi_0(\oL))}\mathcal{T}(p), \\
\overline{K} :=&
\bigotimes_{i=1}^3 \bigotimes_{r \in \overline{W}_i}Q(e^+(r),e^-(r)),
\end{split}
\end{align}
whereby \beq \overline{W}_i = \bigcup_{u=1}^{\on}\pd(\Sigma_i; \pi_0(\ol^u), \pi_0(\uL)). \label{e.w.5} \eeq
Here, we label each $r$ in $\bigcup_{u=1}^{\on}\pd(\Sigma_i; \pi_0(\ol^u), \pi_0(\uL))$, with edges $\{e^+(r),e^-(r)\}$, as explained in Notation \ref{n.q.4}.

We can rearrange $\overline{I} \otimes \overline{J}\otimes \overline{K}$ as \beq \mathcal{Y}_1 \otimes \mathcal{Y}_2 \otimes \mathcal{Y}_3, \label{ex.xy.1} \eeq whereby
\begin{align*}
\mathcal{Y}_i &:= \bigotimes_{p \in \pd(\Sigma_{i};\pi_0(\oL))}\mathcal{R}(p) \otimes\  \bigotimes_{p \in {\rm TDP}(\Sigma_i; \pi_0(\oL))}\mathcal{T}(p)
\bigotimes_{p \in \bigcup_{u=1}^{\on}\pd(\Sigma_i; \pi_0(\ol^u), \pi_0(\uL))} Q(p).
\end{align*}

Now $Q$ commutes with $R$ and $T$. By following the arguments used in the main proof in \cite{CS-Lim02}, together with using Notation \ref{n.l.1}, after time-ordering, we see that we can write Expression \ref{ex.xy.1} as \beq \bigotimes_{p \in U_1}R(p)^{g(e^+(p)), g(\bar{e}^+(p)) }_{g(e^-(p)), g(\bar{e}^-(p)) }\bigotimes_{p \in U_2}T(p)^{g(e^+(p)) }_{g(e^-(p)) }\bigotimes_{p \in U_3}Q(p)^{g(e^+(p)) }_{g(e^-(p)) }, \nonumber \eeq with $U_i$, $i=1,2,3$ defined by Equation (\ref{e.u.1}), and $\mathcal{S} = \overline{T} \times V$. When we apply $\widetilde{\Tr}$ defined in Definition \ref{d.tr.1} to it, the calculations in \cite{CS-Lim02} will gives us Equation (\ref{e.z.2}). This allows us to define a state model for hyperlinks.
\end{proof}

\begin{rem}
Now it is not true that Equation (\ref{e.z.2}) is invariant under diffeomorphism of $\bR^3$, or under time-like isotopy, preserving the time-ordering, as defined in \cite{EH-Lim06}. This is because the $R$-matrices must satisfy the Yang-Baxter equation, as discussed in \cite{CS-Lim02, CS-Lim03}. Hence, it is necessary to restrict the values of $\{q_a\}_{a=1}^{\bar{m}}$.
When the gauge group ${\rm G}_a$ is ${\rm SU}(N)$, we need to restrict $q_a^2$ to be either an integer or half-integer. This will quantized $q_a$'s.

Furthermore, we need to time-order between each pair of component matter loop and component geometric loop. No time-ordering is necessary between a pair of component matter (geometric) loops in $\oL$ ($\uL$).
\end{rem}

\section{Topological Quantum Field Theory}\label{s.tqt}

Suppose we have a closed orientable surface $S \subset \bR^4$, disjoint from the geometric hyperlink $\uL$. We can extend the result from previous section by including the curvature operator. Recall from \cite{EH-Lim05}, we quantized the $\mathfrak{su}(2) \times \mathfrak{su}(2)$-valued curvature operator (summing over repeated indices)
\begin{align}
\begin{split}
F_S(\omega) :=& \frac{1}{2}\int_S \frac{\partial A_{\alpha\beta}^i}{\partial x_0} \otimes dx_0 \wedge dx_i \otimes \hat{E}^{\alpha\beta} + \frac{\partial A_{\alpha \beta}^j}{\partial x_i} \otimes dx_i \wedge dx_j \otimes \hat{E}^{\alpha \beta}  \\
&\ \ \ \ \ \ + A_{\alpha\beta}^iA_{\gamma\mu}^j\otimes dx_i \wedge dx_j \otimes [\hat{E}^{\alpha\beta}, \hat{E}^{\gamma\mu}],
\end{split} \label{e.f.1}
\end{align}
as
\beq \frac{1}{Z_{{\rm EH}}^\kappa}\int_{\omega \in L_\omega^\kappa,\ e \in L_e^\kappa}F_S(\omega) V(\{\ul^v\}_{v=1}^{\un})(e) W(q_0; \{\ol^u, \rho_u\}_{u=1}^{\on})(\omega)  e^{i S_{EH}(e, \omega)}\ D[e] D[\omega], \label{ex.ehc.1} \eeq whereby $V$ and $W$ were defined in Equation (\ref{e.l.7}), $D[e]$ and $D[\omega]$ are non-existent Lebesgue measures on $L_e^\kappa$ and $L_\omega^\kappa$ respectively and $Z_{{\rm EH}}^\kappa$ is a normalization constant given by Equation (\ref{e.l.8}). This path integral was made sense of in \cite{EH-Lim02}, using the Chern-Simon rules.

By taking the limit as $\kappa \rightarrow \infty$, we showed in \cite{EH-Lim05}, that the limit exists and is given by
\begin{align}
\hat{F}_S[Z(q_0; \{\ol^u, \rho_u\}_{u=1}^{\on}, \uL)] :=& -i\sqrt{4\pi}\ {\rm lk}(\uL, S) \otimes \left( \mathcal{E}^+, - \mathcal{E}^- \right)Z(q_0; \{\ol^u, \rho_u\}_{u=1}^{\on}, \uL) \nonumber \\
\equiv &  -i\sqrt{4\pi}\ {\rm lk}(\uL, S) Z(q_0; \{\ol^u, \rho_u\}_{u=1}^{\on}, \uL) \otimes \mathcal{E}. \label{e.ehc.2}
\end{align}
Recall that $\left( \mathcal{E}^+, - \mathcal{E}^- \right) \in \mathfrak{su}(2) \times \mathfrak{su}(2)$.
In defining the linking number between the geometric hyperlink $\uL$ and $S$, we assume that $\pi_a(\uL)$ intersect $S$ finitely many times, $a=0, 1, 2, 3$. Furthermore, $S \subset \bR^4$ is a closed orientable surface and is ambient isotopic to a closed surface in $\{0 \} \times \bR^3$. See \cite{EH-Lim06}.

\begin{thm}\label{t.u.3}
From Expression \ref{ex.toe.1}, we want to compute the limit of the following path integral (indexed by parameter $\kappa$)
\beq
\frac{1}{Z^\kappa}\int_{\Omega^\kappa}F_S(\omega)V(\{\ul^v\}_{v=1}^{\un})(e) W(\vec{q}; \{\ol^u, \varrho_u, \rho_u\}_{u=1}^{\on})(A^{\otau}, \omega)  e^{i S(A^{\otau}, e, \omega)}\ D[A^{\otau},e,\omega]. \label{ex.toe.3} \eeq As $\kappa \rightarrow \infty$, its limit is equal to
\begin{align}
\hat{F}_SZ(\vec{q} ; &\{\ol^\mu, \varrho_\mu, \rho_\mu\}_{\mu \times \overline{T} \times V}, \uL ) \nonumber \\
=& -i\sqrt{4\pi}\ {\rm lk}(\uL, S) Z(\vec{q} ; \{\ol^\mu, \varrho_\mu, \rho_\mu\}_{\mu \in \overline{T} \times V}, \uL  ) \otimes \mathcal{E}. \label{e.toe.4}
\end{align}
Note that $Z(\vec{q} ; \{\ol^\mu, \varrho_\mu, \rho_\mu\}_{\mu \in \overline{T} \times V}, \uL )$ was defined in Equation (\ref{e.z.2}).
\end{thm}

\begin{proof}
Recall
\begin{align*}
W(\omega)
&=\bigotimes_{u=1}^{\on}
\left(
  \begin{array}{cc}
    I_{n_u} &\ 0 \\
    0 &\ \bar{S}(q_0, \ol^u, \rho_u)(\omega) \\
  \end{array}
\right) .
\end{align*}
From Equation (\ref{e.ehc.2}),
\begin{align*}
\tilde{C}_\kappa &:= \frac{1}{Z_{{\rm EH}}^\kappa}
\int_{e \in L_e^\kappa, \omega \in L_\omega^\kappa}F_S(\omega)\otimes \left[\bigotimes_{u=1}^{\on}I_{n_u} \right]^{\otimes^3} V(\{\ul^v\}_{v=1}^{\un})(e)e^{iS_{{\rm EH}}(e,\omega)}\ D[e] D[\omega] \\
&\longrightarrow -i\sqrt{4\pi}\ {\rm lk}(\uL, S)\left[\bigotimes_{u=1}^{\on}I_{n_u}^{\otimes^3}\right] \otimes \mathcal{E},
\end{align*}
as $\kappa \rightarrow \infty$. Refer to Notation \ref{n.n.2}. And for $\bar{S} = \bigotimes_{u=1}^{\on}\bar{S}(q_0, \ol^u, \rho_u)(\omega)$,
\begin{align*}
\tilde{D}_\kappa
&:= \frac{1}{Z_{{\rm EH}}^\kappa}
\int_{e \in L_e^\kappa, \omega \in L_\omega^\kappa} F_S(\omega)\bar{S}^{\otimes^3} V(\{\ul^v\}_{v=1}^{\un})(e)e^{iS_{{\rm EH}}(e,\omega)}\ D[e] D[\omega]  \\
&\longrightarrow -i\sqrt{4\pi}\ {\rm lk}(\uL, S)  \otimes \mathcal{E} \otimes \bigotimes_{u=1}^{\on} \overline{Q}_u,
\end{align*}
as $\kappa \rightarrow \infty$. See Equations (\ref{e.n.1}) and (\ref{e.ehc.2}). Note that
\begin{align}
\overline{Q}_u &= \bigotimes_{k=1}^3 \bigotimes_{p \in \pd(\Sigma_k; \pi_0(\ol^u), \pi_0(\uL))}\exp\left[ -\pi i q_0\ \sigma(p)\rho_u(\mathcal{E})\right].\label{e.q.1}
\end{align}

Therefore,
\begin{align*}
\frac{1}{Z_{{\rm EH}}^\kappa}\int_{e \in L_e^\kappa, \omega \in L_\omega^\kappa} & F_S(\omega) \otimes W(\omega)^{\otimes^3}\ V(\{\ul^v\}_{v=1}^{\un})(e) e^{iS_{{\rm EH}}(e,\omega)}\ D[e] D[\omega] \\
=&
\left(
  \begin{array}{cc}
     \tilde{C}_\kappa  &\ 0 \\
    0 &\ \tilde{D}_\kappa \\
  \end{array}
\right)
\longrightarrow -i\sqrt{4\pi}\ {\rm lk}(\uL, S) \otimes \overline{K}\otimes \mathcal{E} ,
\end{align*}
as $\kappa \rightarrow \infty$. Note that $\overline{K}$ was defined in Equation (\ref{e.d.1}).

From Equation (\ref{e.w.6}) and Expression \ref{ex.toe.2},
\begin{align*}
\widetilde{\Tr}& \bigotimes_{i=1}^3  Y_i \otimes \ \frac{1}{Z_{{\rm EH}}}\int_{\omega, e} F_S(\omega) \otimes W(\omega)^{\otimes^3}V(\{\ul^v\}_{v=1}^{\un})(e) e^{iS_{{\rm EH}}(e,\omega)} D[e]D[\omega] \\
&= \widetilde{\Tr} \bigotimes_{i=1}^3  Y_i^\kappa  \otimes
\left(
  \begin{array}{cc}
     \tilde{C}_\kappa  &\ 0 \\
    0 &\ \tilde{D}_\kappa \\
  \end{array}
\right)\longrightarrow -i\sqrt{4\pi}\ {\rm lk}(\uL, S)\widetilde{\Tr}[\overline{I} \otimes \overline{J} \otimes \overline{K}] \otimes \mathcal{E} ,
\end{align*}
as $\kappa \rightarrow \infty$. Note that $\overline{I}, \overline{J}$ were defined in Equation (\ref{e.d.1}). By Theorem \ref{t.u.1}, we have that $\widetilde{\Tr}[\overline{I} \otimes \overline{J} \otimes  \overline{K}] = Z(\vec{q}; \{\ol^\mu, \varrho_\mu, \rho_\mu\}_{\mu \in \overline{T} \times V}, \uL )$.
\end{proof}

Linking is a topological concept, and the linking number between a non-intersecting closed orientable surface and a loop in $\bR^4$ is defined. See \cite{ Milnor1965}. The definition of a linking number between a hyperlink and a surface is given in \cite{EH-Lim06}, and can be shown to be consistent with the definition of a linking number, given in \cite{Horowitz:1989km}, differing by a factor 1/2.

The above Equation (\ref{e.toe.4}) involves the linking number between a geometric hyperlink and a closed orientable surface in $\bR^4$. Furthermore, we see that it is a scalar multiple of $Z(\vec{q}; \{\ol^\mu, \varrho_\mu, \rho_\mu\}_{\mu \in \overline{T} \times V}, \uL )\mathcal{E}$, hence the Wilson Loop observable for the time-like hyperlink $\chi(\oL, \uL)$, is an eigenfunctional of the quantized curvature operator. If we impose time-ordering between all pairs of component matter loop and component geometric loop, then this expression will include the linking number between component knots, Homfly-type knot invariants, and the linking number between a closed surface and a time-like hyperlink. This expression is invariant under the time-like, preserving time-ordering equivalence relation, as described in \cite{EH-Lim06}.

Can we quantize the curvature $dA + A \wedge A$ using a Chern-Simons action? The answer is yes. Suppose $S \subset \bR^3$ is an orientable compact surface, with or without boundary, and that $\pi_0(\oL)$ intersect $S$ at finitely many points. We can consider the following Chern-Simons path integral \beq \frac{1}{Z_{{\rm CS}}^\kappa}\int_{A \in K_{{\otau}(3)}^{\kappa}} \mathcal{F}_S(A)W(q; \{\pi_0(\ol^u), \varrho_u\}_{u=1}^{\on})(A) e^{iS_{{\rm CS}}(A)} D[A], \label{ex.a.5} \eeq
whereby the curvature operator
\begin{align}
\begin{split}
\mathcal{F}_S(A) :=& \frac{1}{2}\sum_{a=1}^{\bar{m}}\int_S \frac{\partial A_\alpha^{a,1}}{\partial x_3} \otimes dx_3 \wedge dx_1 \otimes E_a^\alpha + \frac{\partial A_\alpha^{a,2}}{\partial x_3} \otimes dx_3 \wedge dx_2 \otimes E_a^\alpha  \\
&+ \left[\frac{\partial A_\alpha^{a,2}}{\partial x_1} - \frac{\partial A_\alpha^{a,1}}{\partial x_2}\right]dx_1 \wedge dx_2\otimes E_a^\alpha  + A_\alpha^{a,1}A_\beta^{a,2}\otimes dx_1 \wedge dx_2 \otimes [E_a^\alpha, E_a^\beta],
\end{split}
\label{e.f.2}
\end{align}
is a $\mathfrak{g}$-valued 2-form.

After applying the Chern-Simons rules to the path integral Expression \ref{ex.a.5}, the limit as $\kappa \rightarrow \infty$, if it exists, will not yield the linking number between a link and a surface in $\bR^3$. In fact, the limit is either 0 or infinity, the latter due to an extra factor of $\sqrt \kappa$. Even for the simplest case of an unknot in $\bR^3$, which intersects a disc in the $x_2-x_3$ plane once, we see that the limit as $\kappa \rightarrow \infty$ for Expression \ref{ex.a.5} will yield 0, when we consider a semi-simple Lie group $G$. Finally, the above path integral Expression \ref{ex.a.5} may not be invariant under diffeomorphism.

\begin{rem}
A quick way to argue is to consider an abelian Chern-Simons theory. For a closed surface $S \subset \bR^3$, we see that $\int_S dA = 0$ using Stokes' Theorem. Thus, a Chern-Simons path integral involving the integral of curvature for a closed surface, will yield 0. This means that the quantized curvature of a closed surface will be zero for an abelian Chern-Simons quantum theory.
\end{rem}

\section{Area and volume operator}\label{s.av}

Let $S$ be an orientable compact surface inside the spatial subspace $\bR^3 \hookrightarrow \bR \times \bR^3$, disjoint from the matter hyperlink $\oL$. Because we can consider ambient isotopy of $S$ in $\bR^3$, we will assume that $S$ is inside $x_2-x_3$ plane. Furthermore, we insist that $\pi_0(\oL)$ intersects the surface $S$ at most finitely many points. Using the dynamical variables $\{B_\mu^a\}$ and the Minkowski metric $\eta^{ab}$, we see that the metric $g^{ab} \equiv B^a_\mu\eta^{\mu\gamma}B^b_\gamma$ and the corresponding
area is given by \beq {\rm Area\ of}\ S(e) := A_S(e) \equiv \int_S \sqrt{g^{22}g^{33} - (g^{23})^2} \ dA. \nonumber \eeq

By using the Chern-Simon rules in \cite{EH-Lim02}, we made sense of the following path integral expression (indexed by a parameter $\kappa$),
\beq \frac{1}{Z_{{\rm EH}}^\kappa}\int_{\omega \in L_\omega^\kappa,\ e \in L_e^\kappa}A_S(e) V(\{\ul^v\}_{v=1}^{\un})(e) W(q_0; \{\ol^u, \rho_u^\pm\}_{u=1}^{\on})(\omega)\ e^{i S_{{\rm EH}}(e, \omega)}\ D[e] D[\omega], \label{ex.eha.1} \eeq $V$ and $W$ as defined in Equation (\ref{e.l.7}) and $Z_{{\rm EH}}^\kappa$ is a normalization constant given by Equation (\ref{e.l.8}). Note that $\{\rho_u^\pm\}_{u=1}^{\on}$ is a representation for the hyperlink $\oL$. In \cite{EH-Lim03}, we took the limit of the path integral expression and quantized the area of $S$ into an operator $\hat{A}_S$,

In \cite{EH-Lim06}, we defined the piercing number $\nu_S(l)$ between a compact surface $S$ in $\bR^3 \equiv \{0\} \times \bR^3$, with or without boundary, and a loop $l$. This is a well-defined invariant, up to time-like isotopy, preserving time-ordering. It counts the number of times $\pi_0(l)$ intersects the surface $S$ in $\bR^3$, the intersection points termed as piercings. We always choose a representative of $l$ and $S$, up to the time-like and time-ordered equivalence relation defined in \cite{EH-Lim06}, such that the $\pi_0(l)$ and $S$ in $\bR^3$ have the minimum number of piercings, which gives us $\nu_S(l)$. This is an invariant under the said equivalence relation.

The main theorem in \cite{EH-Lim03} says that
\begin{align}
\lim_{\kappa \rightarrow \infty}\frac{1}{Z_{{\rm EH}}^\kappa}&\int_{\omega \in L_\omega^\kappa,\ e \in L_e^\kappa}V(\{\ul^v\}_{v=1}^{\un})(e) W(q_0; \{\ol^u, \rho_u^+\}_{u=1}^{\on})(\omega)A_S(e)\  e^{i S_{{\rm EH}}(e, \omega)}\ D[e] D[\omega] \nonumber \\
&= \frac{|q_0|\sqrt\pi}{2} \left[
\sum_{u=1}^{\on}\nu_{S}(\ol^u)  \sqrt{\xi_{\rho_u^+}}\right]Z^+(q_0; \{\ol^u, \rho_u\}_{u=1}^{\on}, \uL), \label{e.a.3a}
\end{align}
and
\begin{align}
\lim_{\kappa \rightarrow \infty}\frac{1}{Z_{{\rm EH}}^\kappa}&\int_{\omega \in L_\omega^\kappa,\ e \in L_e^\kappa}V(\{\ul^v\}_{v=1}^{\un})(e) W(q_0; \{\ol^u, \rho_u^-\}_{u=1}^{\on})(\omega)A_S(e)\  e^{i S_{{\rm EH}}(e, \omega)}\ D[e] D[\omega] \nonumber \\
&= i\frac{|q_0|\sqrt\pi}{2} \left[
\sum_{u=1}^{\on}\nu_{S}(\ol^u)\sqrt{\xi_{\rho_u^-}}\right]Z^-(q_0; \{\ol^u, \rho_u\}_{u=1}^{\on}, \uL). \label{e.a.3b}
\end{align}
See Definition \ref{d.z.1}.

\begin{rem}
Note that in defining these path integrals, it was necessary to partition the surface $S$. This limit was shown to be independent of any partition.
\end{rem}

\begin{thm}\label{t.u.4}
From Expression \ref{ex.toe.1}, we want to compute the limit of the following path integral (indexed by parameter $\kappa$)
\beq
\frac{1}{Z^\kappa}\int_{\Omega^\kappa}A_S(e)V(\{\ul^v\}_{v=1}^{\un})(e) W(\vec{q}; \{\ol^u, \varrho_u, \rho_u^+\}_{u=1}^{\on})(A^{\otau}, \omega)  e^{i S(A^{\otau}, e, \omega)}\ D[A^{\otau}, e, \omega]. \label{ex.as.1} \eeq Its limit, as $\kappa \rightarrow \infty$, allows us to define the area operator $\hat{A}_S$,
\begin{align}
\hat{A}_SZ(\vec{q};& \{\ol^{\mu}, \varrho_\mu, \rho_\mu^+\}_{\mu \in \overline{T}\times V}, \uL ) \nonumber \\
&:= \frac{|q_0|\sqrt\pi}{2}\left[
\sum_{u=1}^{\on}\nu_{S}(\ol^u)  \sqrt{\xi_{\rho_u^+}}\right]\ \prod_{u=1}^{\on}\prod_{i=1}^3\mathcal{W}^+(i,u)
,\label{e.a.1}
\end{align}
whereby \beq \mathcal{W}^\pm(i,u) :=
\Tr_{\rho_u^\pm}\ \prod_{p \in \pd(\Sigma_i; \pi_0(\ol^u), \pi_0(\uL))}\exp\left[ \mp\pi i q_0\ \sigma(p)\mathcal{E}^\pm\right] . \label{e.w.4} \eeq
Similarly, the limit of the following path integral (indexed by parameter $\kappa$)
\beq
\frac{1}{Z^\kappa}\int_{\Omega^\kappa}A_S(e)V(\{\ul^v\}_{v=1}^{\un})(e) W(\vec{q}; \{\ol^u, \varrho_u, \rho_u^-\}_{u=1}^{\on})(A^{\otau}, \omega)  e^{i S(A^{\otau}, e, \omega)}\ D[A^{\otau},e,\omega], \label{ex.as.2} \eeq is equal to
\begin{align}
\hat{A}_SZ(\vec{q};& \{\ol^\mu, \varrho_\mu, \rho_\mu^-\}_{\mu\in \overline{T} \times V}, \uL ) \nonumber \\
&:= i\frac{|q_0|\sqrt\pi}{2} \left[
\sum_{u=1}^{\on}\nu_{S}(\ol^u)\sqrt{\xi_{\rho_u^-}}\right]\ \prod_{u=1}^{\on}\prod_{i=1}^3\mathcal{W}^-(i,u).\label{e.a.2}
\end{align}
\end{thm}

\begin{proof}
Write
\begin{align*}
\tilde{S}_u^+(\omega)
&:= \exp\left[ \frac{q_0}{3}\int_{\ol^u} A^i_{\alpha\beta} \otimes dx_i\otimes \rho_u^+(\hat{E}^{\alpha\beta})\right].
\end{align*}
From Equation (\ref{e.a.3a}),
\begin{align*}
C_\kappa := \frac{1}{Z_{{\rm EH}}^\kappa}&
\int_{e \in L_e^\kappa, \omega \in L_\omega^\kappa} \bigotimes_{u=1}^{\on}I_{n_u}^{\otimes^3}\otimes A_S(e) V(\{\ul^v\}_{v=1}^{\un})(e)e^{i S_{{\rm EH}}(e,\omega)}\ D[e] D[\omega] \\
\longrightarrow& 0,
\end{align*}
as $\kappa \rightarrow \infty$. And for $V(e) \equiv V(\{\ul^v\}_{v=1}^{\un})(e)$,
we have
\begin{align*}
D_\kappa^{+} :=\frac{1}{Z_{{\rm EH}}^\kappa}
& \int_{e \in L_e^\kappa, \omega \in L_\omega^\kappa} A_S(e) \bigotimes_{u=1}^{\on}\tilde{S}_u^+(\omega)^{\otimes^3} V(e)e^{i S_{{\rm EH}}(e,\omega)}\ D[e] D[\omega] \\
\longrightarrow& \frac{|q_0|\sqrt\pi}{2} \left[\sum_{u=1}^{\on}
\nu_{S}(\ol^u)  \sqrt{\xi_{\rho_u^+}}\right]\otimes \bigotimes_{u=1}^{\on}  \overline{Q}_u ,
\end{align*}
as $\kappa \rightarrow \infty$. See Equations (\ref{e.n.1}) and (\ref{e.a.3a}). Note that $\overline{Q}_u$ was defined in Equation (\ref{e.q.1}).

From Expression \ref{ex.toe.2},
\begin{align*}
\widetilde{\Tr}&\ \bigotimes_{i=1}^3  Y_i \otimes \ \frac{1}{Z_{{\rm EH}}}\int_{\omega, e} A_S(e) W^+(\omega)^{\otimes^3}V(e) e^{iS_{{\rm EH}}(e,\omega)} D[e]D[\omega] \\
&= \widetilde{\Tr}\ \frac{1}{Z_{{\rm EH}}}\int_{\omega, e} A_S(e) W^+(\omega)^{\otimes^3}V(e) e^{iS_{{\rm EH}}(e,\omega)} D[e]D[\omega]\ \otimes \bigotimes_{i=1}^3  Y_i \\
&= \widetilde{\Tr}\
\left(
  \begin{array}{cc}
     C_\kappa  &\ 0 \\
    0 &\ D_\kappa^{+} \\
  \end{array}
\right) \otimes
\bigotimes_{i=1}^3  Y_i^\kappa,
\end{align*}
whereby index $\kappa$ is omitted from the expression, and \beq W^+(\omega) := \bigotimes_{u=1}^{\on}
\left(
  \begin{array}{cc}
    I_{n_u} &\ 0 \\
    0 &\ \tilde{S}_u^+(\omega) \\
  \end{array}
\right). \nonumber \eeq
The limit as $\kappa \rightarrow \infty$, is given by
\begin{align*}
\frac{|q_0|\sqrt\pi}{2}& \left[
\sum_{u=1}^{\on}\nu_{S}(\ol^u)  \sqrt{\xi_{\rho_u^+}}\right]\widetilde{\Tr}\
\bigotimes_{u=1}^{\on}\overline{Q}_u \\
&= \frac{|q_0|\sqrt\pi}{2} \left[
\sum_{u=1}^{\on}\nu_{S}(\ol^u)  \sqrt{\xi_{\rho_u^+}}\right]\ \prod_{u=1}^{\on}\prod_{i=1}^3\mathcal{W}^+(i,u)
.
\end{align*} This proves Equation (\ref{e.a.1}). The proof for Equation (\ref{e.a.2}) is similar, hence omitted.
\end{proof}

Fix a compact solid region $R \subset \bR^3$, possibly disconnected with finite number of components, disjoint from the matter hyperlink $\oL$. Its boundary is a closed (compact without boundary) surface. Using the dynamical variables $\{B_\mu^a\}$ and the Minkowski metric $\eta^{ab}$, we see that the metric $g^{ab} \equiv B^a_\mu\eta^{\mu\gamma}B^b_\gamma$ and the corresponding volume $V_R$ is given by \beq V_R(e) := \int_R \sqrt{\epsilon_{ijk}\epsilon_{\bar{i}\bar{j}\bar{k}}g^{i\bar{i}}g^{j\bar{j}}g^{k\bar{k}}}. \nonumber \eeq

By using the Chern-Simon rules, in \cite{EH-Lim02}, we made sense of the following path integral expression (indexed by a parameter $\kappa$),\beq \frac{1}{Z_{{\rm EH}}^\kappa}\int_{\omega \in L_\omega^{\kappa},\ e \in L_e^\kappa}V_R(e)V(\{\ul^v\}_{v=1}^{\un})(e) W(q_0; \{\ol^u, \rho_u\}_{u=1}^{\on})(\omega)\  e^{i S_{{\rm EH}}(e, \omega)}\ D[e] D[\omega], \label{ex.v.11} \eeq whereby $V$ and $W$ were defined in Equation (\ref{e.l.7}) and $Z_{{\rm EH}}^\kappa$ is a normalization constant given by Equation (\ref{e.l.8}). In \cite{EH-Lim04}, we took the limit of the path integral expression and quantized the volume of $R$ into an operator $\hat{V}_R$ using the preceding path integral expression.

In \cite{EH-Lim06}, we defined the confinement number $\nu_R(l)$, which counts the number of nodes in the interior of $R$. The set of nodes on the projected loop $\pi_0(l)$, are in 1-1 correspondence with the set of half twists on the link diagram of a framed knot $\pi_0(l)$, as described in Item \ref{i.f.1} of Definition \ref{d.c.1}. Under the time-like equivalence relation, as defined in \cite{EH-Lim06}, $\nu_R(l)$ is an invariant.

From the main theorem in \cite{EH-Lim04}, we have the following result, which says that
\begin{align}
\lim_{\kappa \rightarrow \infty}\frac{1}{Z_{{\rm EH}}^\kappa}\int_{\omega \in L_\omega^{\kappa},\ e \in L_e^\kappa}&V_R(e)V(\{\ul^v\}_{v=1}^{\un})(e) W(q_0; \{\ol^u, \rho_u^\pm\}_{u=1}^{\on})(\omega)\  e^{i S_{{\rm EH}}(e, \omega)}\ D[e] D[\omega] \nonumber \\
&= \frac{q_0^2\pi^{3/2}}{2}\left[\sum_{u=1}^{\on} \nu_R(\ol^u)
\xi_{\rho_u^\pm} \right]Z^\pm(q_0; \{\ol^u, \rho_u\}_{u=1}^{\on}, \uL). \label{e.v.1}
\end{align}

Note that to define this path integral, it was necessary to partition the solid region $R$. The limit was shown to be independent of any partition. We can now state a similar theorem for the volume operator. The proof is similar to the case for the area operator, hence omitted.

\begin{thm}\label{t.u.2}
From Expression \ref{ex.toe.1}, we want to compute the limit of the following path integral (indexed by the parameter $\kappa$)
\beq
\frac{1}{Z^\kappa}\int_{ \Omega^\kappa}V_R(e)V(\{\ul^v\}_{v=1}^{\un})(e) W(\vec{q}; \{\ol^\mu, \varrho_u, \rho_u^\pm\}_{u=1}^{\on})(A^{\otau}, \omega)  e^{i S(A^{\otau}, e, \omega)}\ D[A^{\otau}, e, \omega]. \label{ex.vs.1} \eeq Its limit, as $\kappa \rightarrow \infty$, allows us to define the volume operator $\hat{V}_R$,
\begin{align*}
\hat{V}_RZ(\vec{q};& \{\ol^\mu, \varrho_\mu, \rho_\mu^\pm\}_{\mu\in \overline{T} \times V}, \uL ) \nonumber \\
&:= \frac{q_0^2\pi^{3/2}}{2}\left[\sum_{u=1}^{\on} \nu_R(\ol^u)
\xi_{\rho_u^\pm} \right]\ \prod_{u=1}^{\on}\prod_{i=1}^3\mathcal{W}^\pm(i,u).
\end{align*}
See Equation (\ref{e.w.4}) for the definition of $\mathcal{W}^\pm(i,u)$.
\end{thm}

Clearly, we see that the Wilson Loop observable $Z(\vec{q}; \{\ol^\mu, \varrho_\mu, \rho_\mu^\pm\}_{\mu\in \overline{T} \times V}, \uL )$, is not an eigenfunctional for both the area and volume operators. In order for the Wilson Loop observable to be an eigenfunctional, it is necessary for $\varrho_u$ to be the trivial representation. In other words, $\varrho_u: \mathfrak{g} \rightarrow 0$.

\begin{rem}
Note that from Equations (\ref{e.n.1}) and (\ref{e.a.4}), even after imposing time-ordering between each pair of matter of geometric loop, \beq
Z^\pm(q_0; \{\ol^u, \rho_u\}_{u=1}^{\on}, \uL) \neq \prod_{u=1}^{\on}\prod_{i=1}^3\mathcal{W}^\pm(i,u) . \nonumber \eeq  Indeed, we will write
\begin{align}
Z(q_0; \{\ol^\mu, 0, \rho_\mu^\pm\}_{\mu\in \overline{T} \times V}, \uL ) &:= \prod_{u=1}^{\on}\prod_{i=1}^3\mathcal{W}^\pm(i,u)
\equiv [Z^\pm(q_0/3; \{\ol^u, \rho_u\}_{u=1}^{\on}, \uL)]^3  \label{e.a.7}
\end{align}
in future. Refer to Equation (\ref{e.w.4}).
\end{rem}

Explicitly, from Expression \ref{ex.toe.1}, we have
\begin{align}
\frac{1}{Z^\kappa}&\int_{(A^{\otau}, \omega, e)\in \Omega^\kappa}V(\{\ul^v\}_{v=1}^{\un})(e) W(\vec{q}; \{\ol^\mu, 0, \rho_\mu\}_{u=1}^{\on})(A^{\otau}, \omega)  e^{i S(A^{\otau}, e, \omega)}\ D[A^{\otau}, e, \omega] \nonumber\\
=& \frac{1}{Z^\kappa_{{\rm EH}}}\int_{\omega , e }V(\{\ul^v\}_{v=1}^{\un})(e)W(q_0/3; \{\ol^u, \rho_u^\pm\}_{u=1}^{\on})(\omega)^3   e^{i S_{{\rm EH}}( e, \omega)}\ D[e]D[\omega] \nonumber \\
&\ \times\frac{\int_{A^{\otau}} e^{iS_{{\rm CS}}(A^{\otau})} D[A^{\otau}]}{Z^\kappa_{{\rm CS}}} \nonumber \\
\longrightarrow& [Z^\pm(q_0/3; \{\ol^u, \rho_u\}_{u=1}^{\on}, \uL)]^3 = Z(q_0; \{\ol^\mu, 0, \rho_\mu^\pm\}_{\mu\in \overline{T} \times V}, \uL ) , \label{e.a.5}
\end{align}
as we take $\kappa \rightarrow \infty$.

\section{Final comments}\label{s.fr}

The Wilson Loop observable, given by Equation (\ref{e.w.3}), defined using a Chern-Simons path integral, is complex-valued. In contrast, the Wilson Loop observable given by Equation (\ref{e.weh.1}), defined using an Einstein-Hilbert path integral, is real-valued. In both path integrals, the common denominator is the colored matter hyperlink $\oL$, i.e. each matter loop component carries a representation for the Lie algebra $\mathfrak{g} \times [\mathfrak{su}(2) \times \mathfrak{su}(2)]$ of ${\rm G} \times [{\rm SU(2)} \times  {\rm SU(2)}]$.

In General Relativity, solving for the Riemannian metric in Einstein's equations will give us the geometry on the ambient space $\bR^4$. The components $\{A_{\alpha\beta}^a\}$ of the $\mathfrak{su}(2) \times \mathfrak{su}(2)$-valued connection $\omega$ in 4-dimensional space, can be interpreted as a Riemannian connection on a tangent bundle in $\bR^4$. It is paramount then that we can give a geometrical meaning to the $\mathfrak{g}_a$-valued connection.

In general, the $\mathfrak{g}_a$-valued connection $A_a$ has no geometric meaning in $\bR^3$. For example, the curvature of an electromagnetic 4-potential, gives us the electromagnetic field tensor, but it has no geometric meaning.\footnote{In \cite{Lindgren_2021}, the authors tried to explain how one can derive electromagnetic theory purely from geometric considerations, using a space-time metric.} This might explain why a quantized curvature using a Chern-Simons path integral, will not yield topological invariants. We also explained in the last paragraph in Section \ref{s.tqt}, quantizing $\mathfrak{g}$-valued curvature using a Chern-Simons path integral, will yield 0 or infinity.

In \cite{EH-Lim02}, we showed that in defining the Einstein-Hilbert path integrals for the area and volume operators, it was necessary to partition the surface and the solid region respectively. The limit of each respective expressions, is independent of this partition. See \cite{EH-Lim03, EH-Lim04}. This gives us the notion of a quanta of area or volume. As such, the area and volume operators should be considered as local operators, unlike the curvature operator, which depends on the global topology of the surface and the hyperlink. Refer to \cite{EH-Lim07} for details.

Hence, LQG can be considered as both a local and global theory. In contrast, the quantum Chern-Simons theory is not a local theory, but a global theory. The nodes on the framed link $\pi_0(\oL)$, which contribute to the `quantized volume' of a solid, viewed as a 3-submanifold with boundary, is 0-dimensional. Therefore, LQG can be called a `0-1-2-3-4' theory, according to \cite{Freed}, as it involves submanifolds of each of the dimension.

As explained earlier on, the representation $\varrho_u$ has to be trivial for each component matter loop, so that the Wilson Loop functional is an eigenfunctional for the area and volume operator. Thus, we no longer have the $R$-matrices given in Notation \ref{n.q.3} and the Wilson Loop observable reduces down to Equation (\ref{e.a.5}). Only the hyperlinking number between the matter hyperlink (representing the particles) and the geometric hyperlink (representing the gravitons), will give us a non-trivial Wilson Loop observable. This will happen at short distances (comparable to Planck distance), when we are able to consider a quanta of area or volume.

As explained in \cite{GUT}, a set of particles is represented as a basis in a vector space $V$, whereby the elements in the Lie algebra, represented as an endomorphism of the vector space $V$, act on it irreducibly. The force carrying bosons are described as representations of generators in this Lie algebra $\mathfrak{g}$.
When the representation of the Lie algebra is trivial, it means the particles, each represented by a component time-like loop in a matter hyperlink, transform trivially under the gauge group ${\rm G}$, and are no longer distinguishable by the fundamental forces in the Standard Model. The gauge groups ${\rm U}(1)$, ${\rm SU}(2)$ and ${\rm SU}(3)$ will no longer act on the particles. We will boldly conclude that at short distances when gravity is dominant, the remaining three fundamental forces are no longer observable.

\end{document}